\documentclass[preprint,11pt]{article}

\usepackage{amssymb,amsfonts,amsmath,amsthm,amscd,dsfont,mathrsfs}
\usepackage{eqsection,graphicx,float,psfrag,epsfig,color,url}
\usepackage{latexsym}

\newtheorem{theorem}{Theorem}[section]

\newtheorem{conj}[theorem]{Conjecture}

\newtheorem{proposition}[theorem]{Proposition}
\newtheorem{corollary}[theorem]{Corollary}
\newtheorem{definition}{Definition}
\newtheorem{remark}{Remark}

\definecolor{Red}{rgb}{1,0,0}
\definecolor{Blue}{rgb}{0,0,1}
\definecolor{Olive}{rgb}{0.41,0.55,0.13}
\definecolor{Green}{rgb}{0,1,0}
\definecolor{MGreen}{rgb}{0,0.8,0}
\definecolor{DGreen}{rgb}{0,0.55,0}
\definecolor{Yellow}{rgb}{1,1,0}
\definecolor{Cyan}{rgb}{0,1,1}
\definecolor{Magenta}{rgb}{1,0,1}
\definecolor{Orange}{rgb}{1,.5,0}
\definecolor{Violet}{rgb}{.5,0,.5}
\definecolor{Purple}{rgb}{.75,0,.25}
\definecolor{Brown}{rgb}{.75,.5,.25}
\definecolor{Grey}{rgb}{.5,.5,.5}
\definecolor{Pink}{rgb}{1,0,1}
\definecolor{DBrown}{rgb}{.5,.34,.16}

\definecolor{Black}{rgb}{0,0,0}

\def\cP{{\cal P}}
\def\cF{{\cal F}}

\def\normal{{\sf N}}

\def\ve{\varepsilon}
\def\prob{{\mathbb P}}

\def\hx{\widehat{x}}
\def\hr{\widehat{r}}
\def\argmin{{\rm argmin}}
\def\seF{{\sf F}}

\def\cost{{\cal C}}

\def\sign{\mbox{\rm sign}}

\def\eps{\varepsilon}

\def\stMSE{{\sf mse}}
\def\laMSE{{\sf MSE}}
\def\MSE{{\rm MSE}}
\def\MMSE{{\rm MMSE}}
\def\sMMSE{\mbox{\rm\footnotesize MMSE}}

\def\reals{{\mathbb R}}
\def\naturals{{\mathbb N}}
\def\<{\langle}
\def\>{\rangle}
\def\E{{\mathbb E}}
\def\identity{{\mathbf I}}

\def\AMP{{\textrm{\rm AMP}}}

\def\tauinf{{\tau_*}}

\def\de{{\rm d}}
\def\LASSO{{\rm LASSO}}

\def\dr{{\delta r}}
\def\htau{\widehat{\tau}}
\def\ttau{\widetilde{\tau}}
\def\xh{\widehat{x}}
\def\sa{{\sf a}}
\def\sb{{\sf b}}
\def\dx{{\delta x}}

\newcommand{\id}{{\rm\bf I}}

\newcommand{\bs}{\backslash}

\def\da{{\partial a}}
\def\di{{\partial i}}
\def\hJ{\widehat{J}}
\def\normeq{\cong}

\begin{document}

\title{Graphical Models Concepts in Compressed Sensing}

\author{Andrea Montanari\thanks{Department of Electrical Engineering and
Department of Statistics, Stanford University}}

\date{}

\maketitle

\begin{abstract}
This paper surveys recent work in applying 
ideas from graphical models and message passing
algorithms to solve large scale regularized regression problems.
In particular, the focus is on compressed sensing reconstruction via
$\ell_1$ penalized least-squares (known as LASSO or BPDN). 
We discuss how to derive fast
approximate message passing algorithms to solve this problem. 
Surprisingly, the analysis of such algorithms allows to 
prove exact high-dimensional limit results for the  LASSO risk.

This paper will appear as a chapter in a book on `Compressed Sensing'
edited by Yonina Eldar and Gitta Kutyniok.
\end{abstract}
%
%
\section{Introduction}\label{sec:intro}

The problem of reconstructing a high-dimensional vector $x\in\reals^n$ 
from a collection of observations
$y\in\reals^m$ arises in a number of contexts, 
ranging from statistical learning
to signal processing.
It is often assumed that the measurement process is approximately 
linear, i.e. that
\begin{eqnarray}
y = Ax+w\, ,\label{eq:FirstModel}
\end{eqnarray}
where $A\in\reals^{m\times n}$ is a known measurement matrix,
and $w$ is a noise vector. 

The graphical models approach to such reconstruction problem 
postulates a joint probability distribution on $(x,y)$ which takes, without 
loss of generality, the form
\begin{eqnarray}
p(\de x,\, \de y) = p(\de y|x)\, p(\de x)\, .\label{eq:Generic}
\end{eqnarray}
The conditional distribution $p(\de y|x)$ models the noise process, 
while the prior $p(\de x)$ encodes information on the vector $x$.
In particular, within compressed sensing, it can describe its sparsity
properties. Within a \emph{graphical models}
approach, either of these distributions (or both) factorizes according
to a specific graph structure. The resulting
posterior distribution $p(\de x|y)$ is used for inferring $x$ given $y$.

There are many reasons to be skeptical about the idea that 
the joint probability distribution $p(\de x,\,\de y)$ 
can be determined, and  used for reconstructing $x$.
To name one such reason for skepticism, 
any finite sample will allow to determine the prior 
distribution of $x$, $p(\de x)$ only within limited accuracy.
A reconstruction algorithm based on the posterior distribution
$p(\de x| y)$ might be sensitive with respect to changes in the prior
thus leading to systematic errors.

One might be tempted to drop the whole approach as a consequence.
We argue that sticking to this point of view is instead fruitful for
several reasons:
\begin{enumerate}
\item \emph{Algorithmic.} Several existing reconstruction methods 
are in fact M-estimators, i.e. they are defined by minimizing an 
appropriate cost function $\cost_{A,y}(x)$ over $x\in\reals^n$
\cite{VanDerWaart}.
Such estimators can be 
derived as Bayesian estimators (e.g. maximum a posteriori probability) 
for specific forms of $p(\de x)$ and $p(\de y|x)$
(for instance by letting $p(\de x| y)\propto 
\exp\{-\cost_{A,y}(x))\, \de x\, $).
The connection is useful both in interpreting/comparing 
different methods, and in adapting known algorithms 
for Bayes estimation.
A classical example of this cross-fertilization
is the paper \cite{NowakEM}.
This review discusses several other examples
in that build on graphical models inference algorithms.
\item \emph{Minimax.} When the prior $p(\de x)$ 
or the noise distributions, and therefore the conditional 
distribution $p(\de y|x)$,  `exist' but
are unknown, it is reasonable to assume that they belong to specific 
structure classes. By this term we refer generically
to a class of probability distributions characterized by a specific
property. For instance,  
within compressed sensing one often assumes that $x$ has at most
$k$ non-zero entries. One can then take $p(\de x)$ 
to be a distribution
supported on $k$-sparse vectors $x \in\reals^n$.
If $\cF_{n,k}$ denotes the class of such distributions, the minimax 
approach strives to achieve the best uniform guarantee over
$\cF_{n,k}$. In other words, the minimax estimator achieves the 
\emph{smallest} expected error (e.g. mean square error) for the `worst' 
distribution in $\cF_{n,k}$. 

It is a remarkable fact in statistical decision theory 
\cite{LehmannCasella} (which follows from a generalization of Von Neumann 
minimax theorem) that
the minimax estimator coincides with the Bayes estimator
for a specific (worst case) prior $p\in\cF_{n,k}$. In one dimension
considerable information is available about the worst case distribution 
and asymptotically optimal estimators (see Section \ref{sec:Scalar}). 
The methods developed here allow to develop similar insights in 
high-dimension.
\item \emph{Modeling.} In some applications it is 
possible to construct fairly accurate models both of the 
prior distribution 
$p(\de x)$ and of the measurement process $p(\de y|x)$.
This is the case for instance in some communications problems, whereby
$x$ is the signal produced by a transmitter (and generated 
uniformly at random according to a known codebook), and 
$w$ is the noise produced by a well-defined physical process
(e.g. thermal noise in the receiver circuitry).
A discussion of some families of practically interesting priors
$p(\de x)$ can be found in \cite{CevherCompressible}.
\end{enumerate} 
Further, the question of modeling the prior in compressed sensing is
discussed from the point of view of Bayesian theory in \cite{Carin}.

The rest of this chapter is organized as follows.
Section \ref{sec:BasicModel} describes a graphical model naturally 
associated to the compressed sensing reconstruction problem. 
Section \ref{sec:Scalar} provides important background on the one-dimensional 
case. Section \ref{sec:MessagePassing} describes a standard message
passing algorithm  --the min-sum algorithm-- and how it can be simplified 
to solve the LASSO optimization problem.
The algorithm is further simplified in Section \ref{sec:AMP}
yielding the AMP algoritm. The analysis of this algorithm is 
outlined in Section \ref{sec:LargeSystem}. As a consequence
of this analysis, it is possible to compute exact high-dimensional
limits for the behavior of the LASSO estimator. Finally in Section
\ref{sec:Generalizations} we discuss a few examples of how the approach 
developed here can be used to address reconstruction problems in which a 
richer structural information is available.
%
%
\subsection{Some useful notation}

Throughout this review, probability measures over 
the real line $\reals$ or the euclidean space $\reals^K$ 
play a special role. It is therefore useful to be careful 
about the probability-theory notation. The less careful reader who 
prefers to pass directly to the `action' is invited to 
skip these remarks at a first reading.

We will use the notation $p$ or
$p(\de x)$ to indicate  probability measures (eventually with subscripts).
Notice that, in the last form, the $\de x$ is only a reminder
of which variable is distributed with measure $p$. 
(Of course one is tempted to think of $\de x$ as an infinitesimal interval
but this intuition is accurate only if $p$ admits a density.)

A special measure (positive, but not normalized and hence not 
a probability measure) is the Lebesgue measure 
for which we reserve the special notation $\de x$
(something like $\mu(\de x)$ would be more consistent but, in 
our opinion, less readable). 
This convention is particularly convenient 
for expressing  in formulae statements of the form   
\emph{`$p$ admits a density $f$ with respect to Lebesgue measure $\de x$,
with $f:x\mapsto f(x) \equiv \exp(-x^2/(2a))/\sqrt{2\pi a}$
a Borel function'}, which we write simply
\begin{eqnarray}
p(\de x) = \frac{1}{\sqrt{2\pi a}}\, e^{-x^2/2a}\, \de x\, .
\end{eqnarray}
It is well known that expectations are defined as integrals with 
respect to the probability measure which we denote as
\begin{eqnarray}
\E_p\{f\} = \int_{\reals}\! f(x) \, p(\de x)\, ,
\end{eqnarray}
sometimes omitting the subscript $p$ in $\E_p$ and $\reals$
in $\int_{\reals}$. 
Unless specified otherwise, we do not assume such probability 
measures to have a density with respect to Lebesgue measure.
The probability measure $p$
is a set function defined on the Borel $\sigma$-algebra,
see e.g. \cite{Billingsley,Williams}.
Hence it makes sense to write $p((-1,3])$ (the probability of the interval
$(-1,3]$ under measure $p$) or $p(\{0\})$ (the probability of the point $0$).
Equally valid would be expressions such as 
 $\de x((-1,3])$ (the Lebesgue measure of $(-1,3]$) 
or $p(\de x)((-1,3])$  (the probability of the interval
$(-1,3]$ under measure $p$) but we avoid them as somewhat clumsy.

A (joint) probability measure over $x\in\reals^K$ and 
$y\in \reals^L$ will be denoted by $p(\de x,\de y)$
(this is just a probability measure over $\reals^{K}\times\reals^L=
\reals^{K+L}$). The corresponding conditional probability measure 
of $y$ given $x$ is denoted by $p(\de x|y)$
(for a rigorous definition we refer to \cite{Billingsley,Williams}).

Finally, we will not make use of cumulative distribution functions 
--commonly called distribution functions in probability theory--
and instead use `probability distribution' interchangeably with
`probability measure'. 

Some fairly standard discrete mathematics notation will also
be useful. The set of first $K$ integers is to be denoted by
$[K] = \{1,\dots,K\}$.
Order of growth of various functions will be characterized 
by the standard `big-$O$' notation. Recall in particular that,
for $M\to\infty$, one writes $f(M) = O(g(M))$ if $f(M)\le C\, g(M)$ 
for some finite constant $C$, $f(M) = \Omega(g(M))$
if $f(M)\ge g(M)/C$ and $f(M) = \Theta(g(M))$ if
$g(M)/C\le f(M)\le Cg(M)$. Further $f(M) = o(g(M))$
if $f(M)/g(M)\to 0$. Analogous notations are used when the argument
of $f$ and $g$ go to $0$.
%
%
\section{The basic model and its graph structure}
\label{sec:BasicModel}

Specifying the conditional distribution of $y$ given 
$x$ is equivalent to specifying the 
distribution of the noise vector $w$. In most of this chapter we shall 
take $p(w)$ to be a Gaussian distribution of mean $0$ and variance
$\beta^{-1}\id$, whence
\begin{eqnarray}
p_{\beta}(\de y|x) = \Big(\frac{\beta}{2\pi}\Big)^{n/2}\, \exp\Big\{
-\frac{\beta}{2}\|y-Ax\|_2^2\Big\}\, \de y\, .
\end{eqnarray}
The simplest choice for the prior consists  in taking 
$p(\de x)$ to be a product distribution with identical
factors $p(\de x) = p(\de x_1)\times\dots\times p(\de x_n)$. 
We thus obtain the joint distribution
\begin{eqnarray}
p_{\beta}(\de x,\, \de y) = \Big(\frac{\beta}{2\pi}\Big)^{n/2}\, \exp\Big\{
-\frac{\beta}{2}\|y-Ax\|_2^2\Big\}\, \de y\, \prod_{i=1}^n
p(\de x_i)\, .\label{eq:GeneralFactorized}
\end{eqnarray}
It is clear at the outset that generalizations 
of this basic model can be 
easily defined, in such a way to incorporate further information 
on the vector $x$ or on the measurement process. 
As an example, consider the case of
block-sparse signals: The index set $[n]$ is partitioned into
blocks $B(1)$, $B(2)$, \dots $B(\ell)$ of equal length $n/\ell$,
and only a small fraction of the blocks is non-vanishing.
This situation can be captured  by assuming that the prior 
$p(\de x)$ factors over blocks. One thus obtains the joint distribution
\begin{eqnarray}
p_{\beta}(\de x,\, \de y) = \Big(\frac{\beta}{2\pi}\Big)^{n/2}\, \exp\Big\{
-\frac{\beta}{2}\|y-Ax\|_2^2\Big\}\, \de y\, \prod_{j=1}^{\ell}
p(\de x_{B(j)})\, ,
\end{eqnarray}
where $x_{B(j)} \equiv (x_{i}:\, i\in B(j))\in\reals^{n/\ell}$.
Other examples of structured priors will be discussed in 
Section \ref{sec:Generalizations}.

The posterior distribution of $x$
given observations $y$ admits an explicit expression, that can be derived from 
Eq.~(\ref{eq:GeneralFactorized}):
\begin{eqnarray}
p_{\beta}(\de x|\, y) = \frac{1}{Z(y)}\, \exp\Big\{
-\frac{\beta}{2}\|y-Ax\|_2^2\Big\}\, \, \prod_{i=1}^n
p(\de x_i)\, ,
\end{eqnarray}
where $Z(y) = (2\pi/\beta)^{n/2}p(y)$ ensures the normalization 
$\int p(\de x|y) = 1$. Let us stress that while this expression is explicit,
computing expectations or marginals of this distribution is a hard 
computational task.

 Finally, the square residuals 
$\|y-Ax\|_2^2$ decompose in a sum of $m$ terms yielding 
\begin{eqnarray}
p_{\beta}(\de x|\, y) = \frac{1}{Z(y)}\, \prod_{a=1}^m\exp\Big\{
-\frac{\beta}{2}\big(y_a-A_a^Tx\big)^2\Big\}\, \, \prod_{i=1}^n
p(\de x_i)\, ,\label{eq:GeneralModelFinal}
\end{eqnarray}
where $A_a$ is the $a$-th row of the matrix $a$. 
This factorized structure is conveniently described 
by a \emph{factor graph}, i.e. a bipartite graph including 
a `variable node' $i\in [n]$ for each variable $x_i$, and a `factor node'
$a\in [m]$ for each term $\psi_a(x)=\exp\{-\beta(y_a-A_a^Tx)^2/2\}$. 
Variable $i$ and factor $a$ are connected by an edge if 
and only if $\psi_a(x)$ depends non-trivially on $x_i$,
i.e. if $A_{ai}\neq 0$. One such factor graph is reproduced in 
Fig.~\ref{fig:FactorGraph}. 
\begin{figure}
\begin{center}
\includegraphics[width=4cm,angle=90]{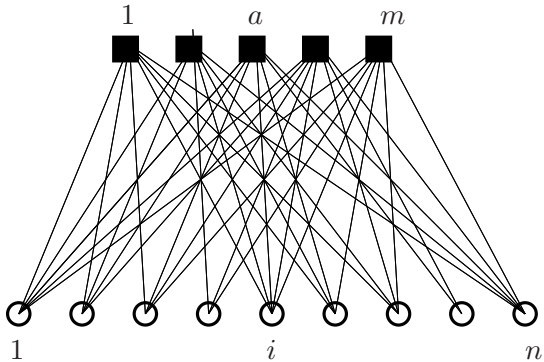}
\put(-200,-12){$1$}
\put(-103,-12){$i$}
\put(-5,-12){$n$}
\put(-158,115){$1$}
\put(-110,115){$a$}
\put(-60,115){$m$}
\end{center}
\caption{{\small Factor graph associated to the probability distribution 
(\ref{eq:GeneralModelFinal}). Empty circles correspond to variables $x_i$,
$i\in [n]$ and squares correspond to measurements $y_a$, $a\in [m]$.}}
\label{fig:FactorGraph}
\end{figure}

An estimate of the signal can be extracted from the 
posterior distribution (\ref{eq:GeneralModelFinal})
in various ways. One possibility is to use conditional expectation
\begin{eqnarray}
\xh_{\beta}(y;p) \equiv \int_{\reals^n} x \; p_{\beta}(\de x|y) \, .
\label{eq:P-Estimator}
\end{eqnarray}
Classically, this estimator is justified by the fact that it achieves
the minimal mean square provided the $p_{\beta}(\de x,\de y)$ 
is the \emph{actual}
joint distribution of $(x,y)$. In the present context
we will not assume that `postulated' prior 
$p_{\beta}(\de x)$ coincides with the actual distribution f $x$, 
and hence $\xh_{\beta}(y;p)$ is not necessarily optimal (with respect
to mean square error).
The best justification for 
$\xh_{\beta}(y;p)$ is that a broad class of estimators 
can be written in the form (\ref{eq:P-Estimator}).

An important problem with the estimator (\ref{eq:P-Estimator})
is that it is in general hard to compute. In order to obtain 
a tractable proxy, we assume that $p(\de x_i)=
p_{\beta,h}(\de x_i) =c\, f_{\beta, h}(x_i)
\, \de x_i$ for $f_{\beta, h}(x_i) = e^{-\beta h(x_i)}$ an un-normalized
probability density function. 
As $\beta$ get large, the integral in Eq.~(\ref{eq:P-Estimator})
becomes dominated by the vector $x$ with the highest 
posterior probability $p_{\beta}$.
One can then replace the integral in 
$\de x$ with a maximization over $x$ and define
\begin{align}
\xh(y;h) &\equiv  
\argmin_{z\in\reals^n} \cost_{A,y}(z;h)\, ,
\label{eq:OptEstimator}\\
\cost_{A,y}(z;h) & \equiv  \frac{1}{2}\|y-Az\|_2^2 +  \sum_{i=1}^n 
h(z_i)\, ,\nonumber
\end{align}
where we assumed for simplicity that $\cost_{A,y}(z;h)$
has a unique minimum.

According to the above discussion, the estimator
$\xh(y;h)$ can be thought of as the $\beta\to\infty$
limit of the general estimator (\ref{eq:P-Estimator}). Indeed, it is
easy to check that, provided $x_i\mapsto h(x_i)$ is upper
semicontinuous, we have
\begin{eqnarray*}
\lim_{\beta\to\infty} \xh_{\beta}(y;p_{\beta, h}) =
 \xh(y;h)\, . 
\end{eqnarray*}
In other words, the posterior mean converges to the mode
of the posterior in this limit. 
Further, $\xh(y;h)$ takes the familiar form of a regression estimator with 
separable regularization. If $h(\,\cdot\,)$ is convex, 
the computation of $\xh$ is tractable. Important 
special cases include $h(x_i) = \lambda x_i^2$, which corresponds to 
ridge regression \cite{HastieBook}, and $h(x_i) = \lambda|x_i|$ which 
corresponds to the LASSO \cite{Tibs96} or basis pursuit denoising (BPDN)
\cite{BP95}. Due to the special role it plays in compressed sensing, we 
will devote special attention to the latter case, that we rewrite explicitly 
below with a slight abuse of notation
\begin{align}
\xh(y) &\equiv 
\argmin_{z\in\reals^n} \cost_{A,y}(z)\, ,
\label{eq:LASSO}\\
\cost_{A,y}(z) & \equiv \frac{1}{2}\|y-Az\|_2^2 +  \lambda
\|z\|_1\, .\nonumber
\end{align}
%
 
%
%
\section{Revisiting the scalar case}
\label{sec:Scalar}

Before proceeding further, it is convenient to 
pause for a moment and consider the special case 
of a single measurement of a scalar quantity, i.e. the case
$m=n=1$.
We therefore have
\begin{eqnarray}
y=x+w\, ,
\end{eqnarray}
and want to estimate $x$ from $y$. 
Despite the apparent simplicity, there exists a copious literature
on this problem with many open problems
\cite{DJHS92,DJ94a,DJ94b,JohnstoneBook}. Here we only want to clarify 
a few points that will come up again in what follows. 

In order to compare various estimators we will assume 
that $(x,y)$ are indeed random variables with some 
underlying probability distribution $p_0(\de x,\de y) = 
p_0(\de x)p_0(\de y|x)$. It is important to stress that this 
distribution is conceptually distinct from
the one used in inference, cf. Eq.~(\ref{eq:P-Estimator}).
In particular we cannot assume to know the actual prior 
distribution of $x$, at least not exactly,
and hence $p(\de x)$ and $p_0(\de x)$  do not coincide. 
The `actual' prior $p_0$ is the distribution of the vector to be inferred,
while the `postulated' prior $p$ is a device used for designing inference 
algorithms.

For the sake of simplicity we also consider Gaussian noise
$w\sim\normal(0,\sigma^2)$ with known noise level
$\sigma^2$.
Various estimators will be compared
with respect to the resulting mean square error
\begin{eqnarray*}
\MSE = \E\{|\xh(y)-x|^2\} = \int_{\reals\times \reals}
|\xh(y)-x|^2\, p_0(\de x,\de y)\, .
\end{eqnarray*}

We can distinguish two  cases:
\begin{itemize}
\item[I.] The signal distribution $p_0(x)$ is known as well. This can 
be regarded as an `oracle' setting. To make contact 
with compressed sensing, we consider distributions that generate
sparse signals, i.e.  that put mass at least $1-\ve$ on $x=0$.
In formulae $p_0(\{0\})\ge 1-\ve$. 
\item[II.] The signal distribution is unknown but it is known
that it is `sparse', namely that it belongs to the class
\begin{eqnarray}
\cF_{\ve}\equiv\big\{\,p_0\, :\;\; p_0(\{0\})\ge 1-\ve\,\big\}\, .
\end{eqnarray}
\end{itemize}

The \emph{minimum mean square error}, is the minimum MSE 
achievable by any estimator $\xh:\reals\to\reals$:
\begin{eqnarray*}
\MMSE(\sigma^2;p_0) = \inf_{\xh:\reals\to\reals} \E\{|\xh(y)-x|^2\}\, .
\label{eq:MMSE_Simple}
\end{eqnarray*}
It is well known that the infimum is achieved 
by the conditional expectation
\begin{eqnarray*}
\xh^{\sMMSE}(y) = \int_\reals x\; p_0(\de x|y)\, .
\end{eqnarray*}
However, this estimator assumes that we are in situation I above, i.e. 
that the prior $p_0$ is known.

In Figure \ref{fig:mmse_vs_st} we plot the resulting MSE 
for a 3 point distribution,
\begin{eqnarray}
p_0 = \frac{\ve}{2}\,\delta_{+1} +(1-\ve)\, \delta_0 +\frac{\ve}{2}\,
\delta_{-1} \, .\label{eq:ThreePoint}
\end{eqnarray}
The MMSE is non-decreasing in $\sigma^2$ by construction, converges to 
$0$ in the noiseless limit $\sigma\to 0$
(indeed the simple rule $\xh(y) =y$ achieves 
MSE equal to $\sigma^2$) and to $\ve$ in the large noise limit
$\sigma\to\infty$ (MSE equal to $\ve$ is achieved by $\xh =0$).
\begin{figure}
\begin{center}
\includegraphics[width=9cm,angle=0]{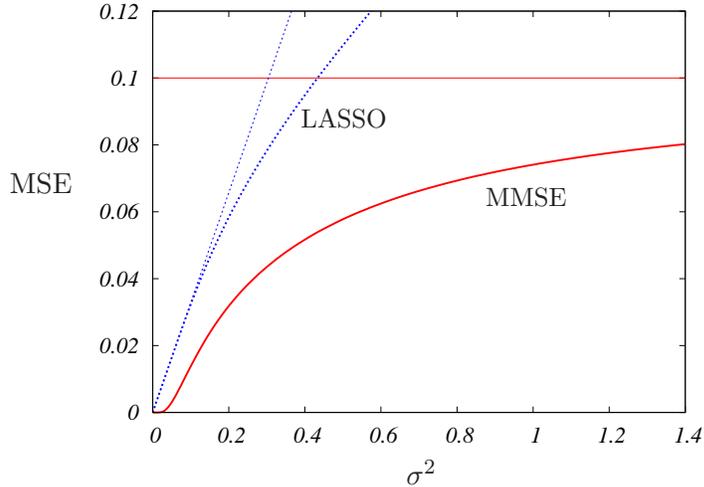}
\put(-120,-10){$\sigma^2$}
\put(-270,100){$\MSE$}
\put(-160,125){{\small LASSO}}
\put(-90,95){{\small MMSE}}
\end{center}
\caption{
{\small Mean square error for estimating a three points random variable,
with probability of non-zero $\ve=0.1$,
in Gaussian noise.
Red line: Minimal mean square error achieved by conditional expectation
(thick) and its large noise asymptote (thin). Blue line: 
Mean square error for LASSO or equivalently
for soft thresholding (thick) and its small noise asymptote 
(thin).}}\label{fig:mmse_vs_st}
\end{figure}

In the more realistic situation II, we do not know the prior
$p_0$. A principled way to deal with this ignorance
would be to minimize the MSE for the worst case 
distribution in the class $\cF_{\ve}$, i.e. to replace the minimization
in Eq.~(\ref{eq:MMSE_Simple}) with the following minimax problem
\begin{eqnarray}
\inf_{\xh:\reals\to\reals}\sup_{p_0\in\cF_{\ve}}\E\{|\xh(y)-x|^2\}\, .\label{eq:MinimaxProblem}
\end{eqnarray}
A lot is known about this problem \cite{DJHS92,DJ94a,DJ94b,JohnstoneBook}.
In particular general statistical decision theory 
\cite{LehmannCasella,JohnstoneBook} implies that the optimum estimator is 
just the posterior expectation for a specific worst case
prior. Unfortunately, even a superficial discussion of
this literature goes beyond the scope of the present review.

Nevertheless, an interesting 
exercise (indeed not a trivial one)
is to consider the LASSO estimator (\ref{eq:LASSO}), which in this
case reduces to 
\begin{eqnarray}
\xh(y;\lambda) = \argmin_{z\in\reals}\Big\{ 
\frac{1}{2}(y-z)^2 +\lambda\, |z|\Big\}\, .\label{eq:OneDLASSO}
\end{eqnarray}
Notice that this estimator is insensitive to the details of the 
prior $p_0$. Instead of the full minimax problem (\ref{eq:MinimaxProblem}),
one can then simply optimize the MSE over $\lambda$.

The one-dimensional optimization problem (\ref{eq:OneDLASSO})
admits an explicit solution
in terms of the  \emph{soft thresholding function}
$\eta:\reals\times\reals_+\to\reals$
defined as follows
\begin{eqnarray}
\label{eq:eta-def}
\eta(y;\theta) = \left\{
\begin{array}{ll}
y-\theta & \mbox{ if $y>\theta$,}\\
0 &  \mbox{ if $-\theta\le y\le\theta$,}\\
y+\theta & \mbox{ if $y<-\theta$.}
\end{array}\right.
\end{eqnarray}
The \emph{threshold} value $\theta$ has to be chosen equal
to the regularization parameter $\lambda$ yielding the simple
solution 
\begin{eqnarray}
\xh(y;\lambda) = \eta(y;\theta)\, ,\;\;\;\;
\mbox{for $\lambda=\theta$}\, .
\end{eqnarray}
(We emphasize the identity of $\lambda$ and $\theta$ in the scalar case, 
because it breaks down in the vector case.)

How should the parameter $\theta$ (or equivalently $\lambda$) be fixed?
The rule is conceptually simple: $\theta$ should minimize the maximal 
mean square error for the class $\cF_{\ve}$. 
Remarkably this complex saddle point problem can be solved 
rather explicitly. The key remark is that the 
worst case distribution over the class $\cF_{\ve}$ can be identified 
and takes the form $p^{\#} = (\ve/2)\delta_{+\infty}+(1-\ve)\delta_0+
(\ve/2)\delta_{-\infty}$
\cite{DJ94a,DJ94b,JohnstoneBook}.

Let us outline how the solution follows from this key fact. 
First of all, it makes sense
to scale $\lambda$ as the noise standard deviation,
because the estimator is supposed to filter out the noise.
We then let $\theta =\alpha\sigma$. 
In Fig.~\ref{fig:mmse_vs_st} we plot the resulting MSE 
when $\theta = \alpha\sigma$, with $\alpha \approx 1.1402$. 
We 
denote the LASSO/soft thresholding mean square error by 
$\stMSE(\sigma^2;p_0,\alpha)$ when the noise variance 
is $\sigma^2$, $x\sim p_0$,  and the regularization
parameter is $\lambda=\theta =\alpha\sigma$. The worst case mean square error
is given by $\sup_{p_0\in\cF_{\ve}}\stMSE(\sigma^2;p_0,\alpha)$.
Since the class $\cF_{\ve}$ is invariant by rescaling, this worst case 
MSE must be proportional to the only scale in the problem, i.e., 
$\sigma^2$. We get
\begin{eqnarray}
\sup_{p_0\in\cF_{\ve}}\stMSE(\sigma^2;p_0,\alpha) = M(\ve,\alpha)\sigma^2\, .
\label{eq:M1def}
\end{eqnarray}
The function $M$ can be computed explicitly by evaluating
the mean square error on the worst case distribution $p^{\#}$
\cite{DJ94a,DJ94b,JohnstoneBook}. 
A straightforward calculation (see also 
\cite[Supplementary Information]{DMM09}, and \cite{NSPT}) yields
\begin{eqnarray}
M(\ve,\alpha) =   \eps\, (1 + \alpha^2) + (1-\eps)[2(1+\alpha^2)\, \Phi(-\alpha)-2\alpha\,\phi(\alpha)]
\end{eqnarray}
where $\phi(z) = e^{-z^2/2}/\sqrt{2\pi}$ is the Gaussian density 
and $\Phi(z) = \int_{-\infty}^z\phi(u)\, \de u$ is the Gaussian distribution.
It is also not hard to show that that $M(\ve,\alpha)$ is the
slope of the soft thresholding MSE at $\sigma^2=0$ 
in a plot like the one in Fig.~\ref{fig:mmse_vs_st}.

\begin{figure}
\includegraphics[width=8cm,angle=0]{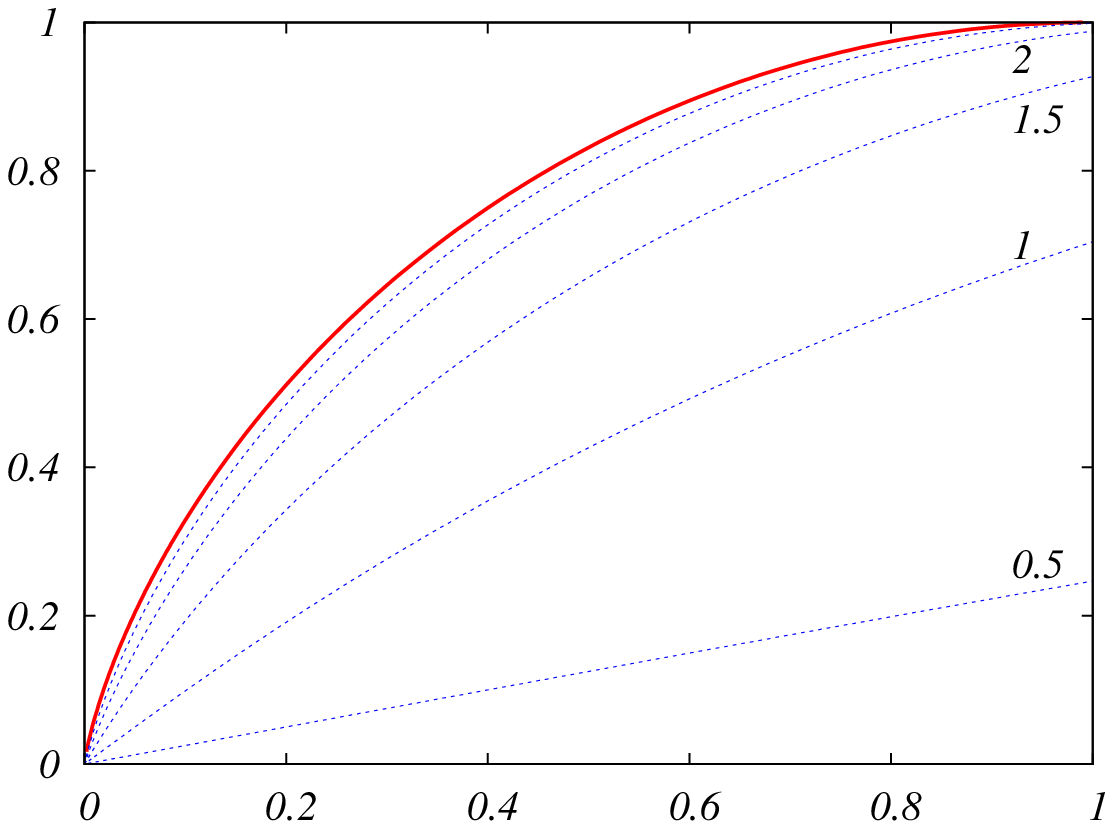}\phantom{AA}
\includegraphics[width=8cm,angle=0]{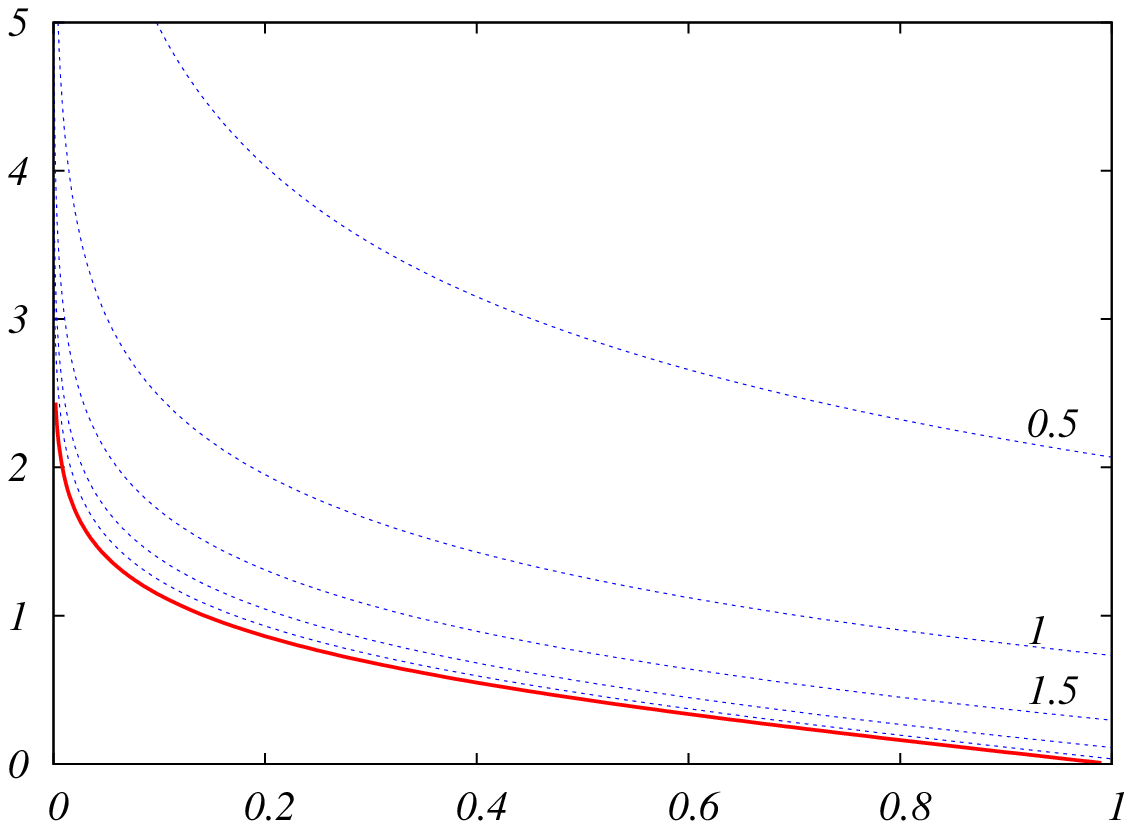}
\put(-350,-7){$\ve$}
\put(-110,-7){$\ve$}
\put(-488,80){$M^{\#}(\ve)$}
\put(-244,80){$\alpha^{\#}(\ve)$}
\caption{{\small Left frame (red line): minimax mean square error under soft 
thresholding for estimation of $\ve$-sparse random variable in 
Gaussian noise. Blue lines correspond to signals of bounded second moment
(labels on the curves refer to the maximum allowed value of
$[\int x^2\, p_0(\de x)]^{1/2}$).
Right frame (red line): Optimal threshold level for the same estimation 
problem. Blue lines again refer to the case of bounded second moment.}}
\label{fig:MmaxScalar}
\end{figure}
Minimizing the above expression over $\alpha$,
we  obtain the soft thresholding minimax risk, and 
the corresponding optimal threshold value
\begin{eqnarray}
M^\#(\ve) \equiv \min_{\alpha\in\reals_+} M(\ve,\alpha)\, ,
\;\;\;\;\;\;\;\;
\alpha^\#(\ve) \equiv \arg\min_{\alpha\in\reals_+} M(\ve,\alpha)\, .
\label{eq:Mhashdef}
\end{eqnarray}
The functions $M^{\#}(\ve)$ and $\alpha^\#(\ve)$ are plotted 
in Fig.~\ref{fig:MmaxScalar}. For comparison we also plot the analogous
functions when the class $\cF_{\ve}$ is replaced by
$\cF_{\ve}(a) = \{p_0\in\cF_{\ve}:\, \int x^2\, p_0(\de x)\le a^2\}$
of sparse random variables with bounded second moment.
Of particular interest is the behavior of these curves in the
very sparse limit $\ve\to 0$,
\begin{eqnarray}
M^{\#}(\ve) = 2\ve \log(1/\ve)\,\big\{1+o(1)\big\}\, ,\;\;\;\;\;
\alpha^{\#}(\ve) =  \sqrt{2\log(1/\ve)}\,
\big\{1+o(1)\big\}\, .\;\;\;\;\label{eq:AsymptoticMSE_ST}
\end{eqnarray}
\begin{figure}
\phantom{a}\hspace{4cm}\includegraphics[width=8cm,angle=0]{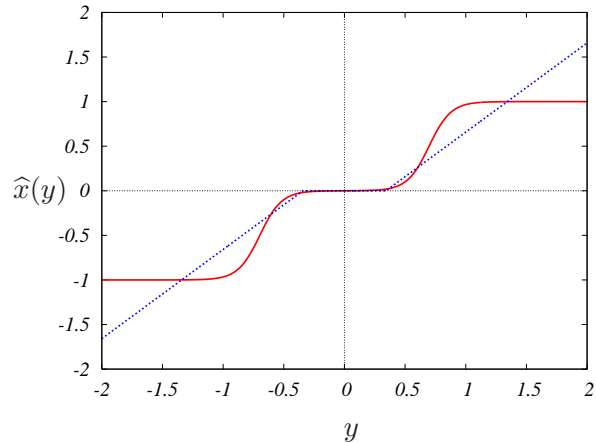}
\put(-105,-10){$y$}
\put(-230,80){$\xh(y)$}
\caption{{\small Red line: The MMSE estimator for the three-point
distribution (\ref{eq:ThreePoint}) with $\ve = 0.1$,
when the noise has 
standard deviation $\sigma = 0.3$.
Blue line: the minimax soft threshold estimator 
for the same setting. The corresponding mean square errors 
are plotted in Fig.~\ref{fig:mmse_vs_st}.}}
\label{fig:Est}
\end{figure}
Getting back to Fig.~\ref{fig:mmse_vs_st}, the reader will notice that
there is a significant gap between the minimal MSE and the MSE
achieved by soft-thresholding. This is the price paid by using
an estimator that is \emph{uniformly good} over the class 
$\cF_{\ve}$ instead of one that is tailored for the distribution $p_0$
at hand. Figure~\ref{fig:Est} compares the two estimators for 
$\sigma=0.3$.
One might wonder whether \emph{all} this price has to be paid,
i.e. whether we can reduce the gap by using a more complex
function instead of the soft threshold $\eta(y;\theta)$.
The answer is both yes and no. On one hand, there 
exist provably superior --in minimax sense--
estimators over $\cF_{\ve}$. Such estimators are of course more complex
than simple soft thresholding. On the other hand,
better estimators have the 
same minimax risk 
$M^{\#}(\ve)=(2\log(1/\ve))^{-1}\,\big\{1+o(1)\big\}$ in 
the very sparse limit, i.e. they improve only the $o(1)$
term as $\ve\to 0$ \cite{DJ94a,DJ94b,JohnstoneBook}.
%
%
\section{Inference via message passing}
\label{sec:MessagePassing}

The task of extending the theory of the previous section
to the vector case (\ref{eq:FirstModel}) might appear
daunting. It  turns out that such extension is instead possible 
in specific high-dimensional limits. The key step 
consists in introducing an appropriate message passing algorithm to
solve the optimization problem (\ref{eq:LASSO}) and then 
analyzing its behavior.
%
%
\subsection{The min-sum algorithm}

\begin{figure}
\begin{center}
\includegraphics[width=4cm,angle=90]{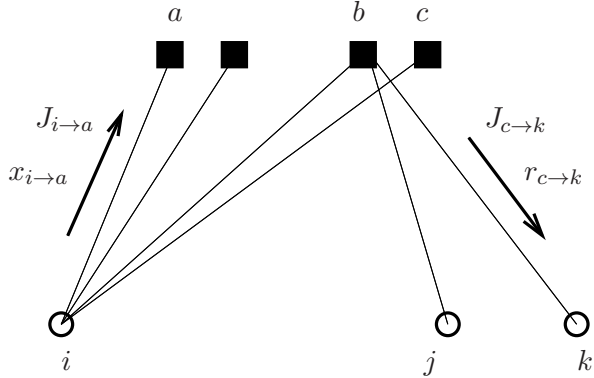}
\put(-200,-12){$i$}
\put(-63,-12){$j$}
\put(-5,-12){$k$}
\put(-160,120){$a$}
\put(-90,120){$b$}
\put(-66,120){$c$}
\put(-210,80){$J_{i\to a}$}
\put(-220,60){$x_{i\to a}$}
\put(-40,80){$J_{c\to k}$}
\put(-25,60){$r_{c\to k}$}
\end{center}
\caption{{\small A portion of the factor graph from
Fig.~\ref{fig:FactorGraph} with notation for messages.}}
\label{fig:FactorMess}
\end{figure}
We start by considering the min-sum algorithm. Min-sum is a 
popular optimization algorithm for graph-structured 
cost functions (see for instance \cite{Pearl,Jordan,MezardMontanari,VanRoy} and 
references therein). 
In order to introduce the algorithm,
we consider a general cost 
function over $x=(x_1,\dots,x_n)$,
that decomposes according to a factor graph as the one shown in 
Fig.~\ref{fig:FactorGraph}:
\begin{eqnarray}
\cost(x) = \sum_{a\in F} \cost_a(x_{\da}) +\sum_{i\in V}\cost_i(x_i)\, .
\label{eq:Decomposable}
\end{eqnarray}
Here $F$ is the set of $m$ \emph{factor nodes} 
(squares in Fig.~\ref{fig:FactorGraph})
and  $V$ is the set of $n$ \emph{variable nodes} (circles in the same figure).
Further $\da$ is the set of neighbors of node $a$ and 
$x_{\da}= (x_i\, :\, i\in\da)$. The min-sum algorithm 
is an iterative algorithm of the belief-propagation type. Its basic variables
are messages: a message is associated to each directed edge in the underlying
factor graph. In the present case, messages are functions on the optimization
variables, and we will denote them as $J_{i\to a}^t(x_i)$
(from variable to factor),
$\hJ_{a\to i}^t(x_i)$ (from factor to variable), with $t$ indicating the 
iteration number. 
Figure \ref{fig:FactorMess} describes the association of
messages to directed edges in the factor graph.
Messages are meaningful up to an additive constant, 
and therefore we will use the special symbol $\normeq$ to denote identity
up to an  additive constant independent of the argument $x_i$.
At the $t$-th iteration they are updated as 
follows\footnote{The reader will notice that for a dense matrix $A$, 
$\di = [n]$ and $\da=[m]$. We will nevertheless stick to the more
general notation, since it is somewhat more transparent.}
\begin{eqnarray}
J^{t+1}_{i\to a}(x_i) & \normeq &\cost_i(x_i) + \sum_{b\in\di\setminus a}
\hJ^t_{b\to i}(x_i)\, ,\label{eq:MinSum1}\\
\hJ^t_{a\to i}(x_i) & \normeq &
\min_{x_{\da\setminus i}}\Big\{\cost_a(x_{\da})+ \sum_{j\in\da\setminus i}
J^t_{j\to a}(x_j)\Big\}\, .\label{eq:MinSum2}
\end{eqnarray}
Eventually, the optimum is approximated by
\begin{eqnarray}
\xh^{t+1}_i &=& \arg\min_{x_i\in\reals} J^{t+1}_{i}(x_i)\, ,\\
J^{t+1}_{i}(x_i) & \normeq &\cost_i(x_i) + \sum_{b\in\di}
\hJ^t_{b\to i}(x_i)\, 
\end{eqnarray}
There exists a vast literature justifying the use of 
algorithms of this type, applying them on concrete problems,
and developing modifications of the basic iteration with better
properties 
\cite{Pearl,Jordan,MezardMontanari,VanRoy,WainwrightJordan,KollerFriedman}.
Here we limit ourselves to recalling that the iteration 
(\ref{eq:MinSum1}), (\ref{eq:MinSum2}) can be regarded as a dynamic 
programming iteration that computes the minimum cost when the 
underlying graph is a tree. Its application to loopy graphs 
(i.e., graphs with closed loops) is
not generally guaranteed to converge.

At this point we notice that the LASSO cost function
Eq.~(\ref{eq:LASSO}) can be decomposed as
in Eq.~(\ref{eq:Decomposable}),
\begin{eqnarray}
\cost_{A,y}(x) & \equiv \frac{1}{2}\sum_{a\in F}(y_a-A_a^Tx)^2 +  \lambda
\sum_{i\in V}|x_i|\, .
\end{eqnarray}
The min-sum updates read
\begin{eqnarray}
J^{t+1}_{i\to a}(x_i) & \normeq &\lambda|x_i| + \sum_{b\in\di\setminus a}
\hJ^t_{b\to i}(x_i)\, ,\label{eq:LassoMinSum1}\\
\hJ^t_{a\to i}(x_i) & \normeq &
\min_{x_{\da\setminus i}}\Big\{\frac{1}{2}(y_a-A_a^Tx)^2
+ \sum_{j\in \da\setminus i}
J^t_{j\to a}(x_j)\Big\}\, .\label{eq:LassoMinSum2}
\end{eqnarray}
%
%
%
\subsection{Simplifying min-sum by quadratic approximation}

Unfortunately, an exact 
implementation of the min-sum iteration appears extremely
difficult because it requires to keep track of $2mn$ messages, each being
a function on the real axis. A possible approach consists in developing
\emph{numerical} 
approximations to the messages. This line of research was initiated 
in \cite{Baraniuk}.
 
Here we will overview an alternative approach that consists 
in deriving \emph{analytical} approximations
\cite{DMM09,DMM_ITW_I,NSPT}. Its advantage is
that it leads to a remarkably simple algorithm, which will be discussed 
in the next section. In order to justify this algorithm 
we will first derive a simplified message passing algorithm,
whose messages are simple real numbers (instead of functions),
and then (in the next section) reduce the number of messages from 
$2mn$ to $m+n$.

Throughout the derivation we shall assume that the 
matrix $A$ is normalized in such a way that its columns 
have zero mean and unit $\ell_2$ norm.
 Explicitly,
we have $\sum_{a=1}^mA_{ai} = 0$ and $\sum_{a=1}^mA_{ai}^2 = 1$. 
In fact it is only sufficient that these conditions are satisfied 
asymptotically for large system sizes. Since however we are only
presenting a heuristic argument, we defer a precise formulation of
this assumption until Section \ref{sec:StateEvolution}.
We also assume that its entries have roughly the same magnitude
$O(1/\sqrt{m})$. Finally, we assume that 
$m$ scales linearly with $n$.
These assumptions are verified by many examples of 
sensing matrices in compressed sensing, e.g. random matrices
with i.i.d. entries or random Fourier sections.
Modifications of the basic algorithm that cope with strong violations
of these assumptions are discussed in 
\cite{BayatiLogReg}.

It is easy to see by induction that the messages $J_{i\to a}^t(x_i)$,
$\hJ^t_{a\to i}(x_i)$ remain, for any $t$, convex functions,
provided they are initialized as convex functions at $t=0$. 
In order to simplify the min-sum equations, we will approximate them by 
quadratic functions. Our first step 
consists in noticing that, as a consequence of
Eq.~(\ref{eq:LassoMinSum2}), the function  $\hJ^t_{a\to i}(x_i)$
depends on its argument only through the combination 
$A_{ai}x_i$. Since $A_{ai}\ll 1$, we can approximate this dependence
through a Taylor expansion (without loss of generality
setting $\hJ_{a\to i}^t(0)=0$):
\begin{eqnarray}
\hJ_{a\to i}^t(x_i) \normeq -\alpha^t_{a\to i}(A_{ai}x_i)+
\frac{1}{2}\beta^t_{a\to i}(A_{ai}x_i)^2+O(A_{ai}^3x_i^3)\, .
\label{eq:hJmess}
\end{eqnarray}
The reason for stopping this expansion at third order should become clear in
a moment.
Indeed substituting in Eq.~(\ref{eq:LassoMinSum1}) we get
\begin{align}
J_{i\to a}^{t+1}(x_i) \normeq \lambda|x_i|-
\Big(\sum_{b\in\di\setminus a}A_{bi}\alpha^t_{b\to i}\Big)\,x_i+
\frac{1}{2}
\Big(\sum_{b\in\di\setminus a}
A_{bi}^2\beta^t_{a\to i}\Big)x_i^2+O(nA_{\cdot i}^3x_i^3)\, .
\label{eq:Jmess1}
\end{align}
Since $A_{ai}= O(1/\sqrt{n})$, the last term is negligible.
At this point we want to approximate $J_{i\to a}^{t}$ by its second 
order Taylor expansion around its minimum. The reason for this is that
only this order of the expansion matters when plugging these messages
in Eq.~(\ref{eq:LassoMinSum2}) to compute $\alpha_{a\to i}^{t}$,
$\beta^t_{a\to i}$. We thus define the quantities
$x^t_{i\to a}$, $\gamma_{i\to a}^t$  as parameters of this 
Taylor expansion:
\begin{eqnarray}
J^t_{i\to a}(x_i) \normeq \frac{1}{2\gamma_{i\to a}^t}(x_i-x_{i\to a}^t)^2
+O((x_i-x_{i\to a}^t)^3)\, .\label{eq:Jmess2}
\end{eqnarray}
Here we include also the case in which the minimum of $J^t_{i\to a}(x_i)$
is achieved at $x_i=0$ (and hence the function is not differentiable
at its minimum) by letting $\gamma^t_{i\to a}=0$ in that case.
Comparing Eqs.~(\ref{eq:Jmess1}) and (\ref{eq:Jmess2}),
and recalling the definition of $\eta(\,\cdot\,;\,\cdot\,)$,
cf. Eq.~(\ref{eq:eta-def}), we get
\begin{eqnarray}
x^{t+1}_{i\to a} = \eta(\sa_1;\sa_2)\, ,\;\;\;\;\;\;\;\;\;
\gamma^{t+1}_{i\to a} = \eta'(\sa_1;\sa_2)\, ,\label{eq:MessPass1}
\end{eqnarray}
where $\eta'(\,\cdot\,;\,\cdot\,)$ denotes the derivative of $\eta$
with respect to its first argument and
we defined
\begin{align}
\sa_1 \equiv \frac{\sum_{b\in\di\setminus a}A_{bi}\alpha^t_{b\to i}}{\sum_{b\in\di\setminus a}A^2_{bi}\beta^t_{b\to i}}\, ,
\;\;\;\;\;\;
\sa_2\equiv \frac{\lambda}{\sum_{b\in\di\setminus a}A^2_{bi}\beta^t_{b\to i}}\, .
\end{align}
Finally, by plugging the parametrization (\ref{eq:Jmess2})
in Eq.~(\ref{eq:LassoMinSum2}) and comparing with Eq.~(\ref{eq:hJmess}),
we can compute the parameters $\alpha^t_{a\to i}$, $\beta^t_{a\to i}$.
A long but straightforward calculation yields
\begin{eqnarray}
\alpha^{t}_{a\to i} & = & \frac{1}
{1+\sum_{j\in\da\setminus i}A_{aj}^2\gamma^t_{j\to a}}
\Big\{y_a-\sum_{j\in\da\setminus i}A_{aj}x^{t}_{j\to a}\Big\}\, ,\\
\beta^{t}_{a\to i} & = & \frac{1}
{1+\sum_{j\in\da\setminus i}A_{aj}^2\gamma^t_{j\to a}}\, .
\label{eq:MessPass4}
\end{eqnarray}
Equations (\ref{eq:MessPass1}) to (\ref{eq:MessPass4}) define a
message passing algorithm that is considerably simpler than the original 
min-sum algorithm: each message consists of a pair of real numbers,
namely $(x^t_{i\to a},\gamma^t_{i\to a})$ for variable-to-factor
messages and $(\alpha_{a\to i},\beta_{a\to i})$ for factor-to-variable
messages. In the next section we will simplify it further and construct an 
algorithm (AMP) with several interesting properties.
Let us pause a moment for making two observations:
\begin{enumerate}
\item The soft-thresholding operator that played an important role
in the scalar case, cf. Eq.~(\ref{sec:Scalar}), reappeared in
Eq.~(\ref{eq:MessPass1}). Notice however that the threshold value that
follows as a consequence of our derivation is not the naive one, namely 
equal to the regularization parameter $\lambda$, but rather a rescaled
one.
\item Our derivation leveraged on the assumption that
the matrix entries $A_{ai}$ are all of the same order, namely
$O(1/\sqrt{m})$. It would be interesting to repeat the above derivation
under different assumptions on the sensing matrix.
\end{enumerate}
%
%
\section{Approximate message passing}
\label{sec:AMP}
The algorithm derived above is still complex in that its memory
requirements scale proportionally to the \emph{product} of
the number of dimensions of the signal and of the number of 
measurements. Further, its computational
complexity per iteration scales quadratically as well.
In this section we will introduce a simpler algorithm,
and subsequently discuss its derivation from the 
one in the previous section.

\subsection{The AMP algorithm, some of its properties, \dots}

The AMP (for approximate message passing) algorithm 
is parameterized by two sequences of scalars:
the thresholds $\{\theta_t\}_{t\ge 0}$ and the `reaction terms'
$\{\sb_t\}_{t\ge 0}$. Starting with initial condition
$x^0 = 0$, it
constructs a sequence of estimates
$x^t\in\reals^n$, and residuals $r^t\in\reals^m$,
according to the following iteration
\begin{eqnarray}
x^{t+1}&=\eta(x^t + A^Tr^t\, ;\theta_t),\label{eq:dmm1}\\
r^t &= y - Ax^t+\sb_t\, r^{t-1}\, ,\label{eq:dmm2}
\end{eqnarray}
for all $t\ge 0$ (with convention $r^{-1}=0$). 
Here and below, given a scalar function
$f:\reals\to\reals$, and a vector $u\in\reals^\ell$,
we adopt the convention of denoting by $f(u)$ the vector
$(f(u_1),\dots,f(u_\ell))$.

The choice of parameters  $\{\theta_t\}_{t\ge 0}$ and
$\{\sb_t\}_{t\ge 0}$ is tightly constrained by the connection with
the min-sum algorithm, as it will be discussed below, but 
the connection with the LASSO is more general. Indeed,
as formalized by the proposition below, general sequences
$\{\theta_t\}_{t\ge 0}$ and
$\{\sb_t\}_{t\ge 0}$ can be used as far as $(x^t,z^t)$ converges.
\begin{proposition}\label{propo:Easy}
Let $(x^*,r^*)$ be a  fixed point of the iteration 
(\ref{eq:dmm1}),
(\ref{eq:dmm2}) for $\theta_t=\theta$, $\sb_t=\sb$ fixed. Then
$x^*$ is a minimum of the $\LASSO$ cost function
(\ref{eq:LASSO}) for 
\begin{eqnarray}
\lambda = \theta(1-\sb)\, .\label{eq:GeneralLambdaTheta}
\end{eqnarray}
\end{proposition}
\begin{proof}
From Eq.~(\ref{eq:dmm1}) we get the fixed point condition
\begin{eqnarray}
x^*+\theta v = x^*+A^Tr^*\, ,
\end{eqnarray}
for $v\in\reals^n$ such that $v_i=\sign(x^*_i)$ if $x^*_i\neq 0$
and $v_i\in[-1,+1]$ otherwise.
In other words, $v$ is a subgradient of the $\ell_1$-norm
at $x^*$, $v\in \partial\|x^*\|_1$. Further from Eq.~(\ref{eq:dmm2})
we get $(1-\sb)r^* = y-Ax^*$. Substituting in the above equation, we get
\begin{eqnarray*}
\theta(1-\sb)v^* = A^T(y-Ax^*)\, ,
\end{eqnarray*}
which is just the stationarity condition for 
the LASSO cost function if $\lambda=\theta(1-\sb)$.
\end{proof}

As a consequence of this proposition, if we find  sequences 
$\{\theta_t\}_{t\ge 0}$, $\{\sb_t\}_{t\ge 0}$ that converge,
and such that the estimates $x^t$ converge as well,
then we are guaranteed that the limit is a LASSO optimum.
The connection with the message passing min-sum algorithm
(see Section \ref{sec:Derivation}) implies an unambiguous prescription for $\sb_t$:
\begin{eqnarray}
\sb_t = \frac{1}{m}\, \|x^t\|_0\, ,\label{eq:CoeffB}
\end{eqnarray}
where $\|u\|_0$ denotes the $0$ pseudo-norm of vector $u$,
i.e. the number of its non-zero components.
The choice of the sequence of thresholds $\{\theta_t\}_{t\ge 0}$
is somewhat more flexible. Recalling the discussion of the
scalar case, it appears to be a good choice to use 
$\theta_t = \alpha\tau_t$ where $\alpha>0$ and $\tau_t$
is the root mean square error of the un-thresholded
estimate $(x^t+A^Tr^t)$. It can be shown that the latter is (in an 
high-dimensional
setting) well approximated by $(\|r^t\|_2^2/m)^{1/2}$. We thus obtain
the prescription
\begin{eqnarray}
\theta_t = \alpha\htau_t\, ,\;\;\;\;\;\;\;\;\;\;\;
\htau_t^2 = \frac{1}{m}\, \|r^t\|_2^2\, .\label{eq:ThresholdChoice1}
\end{eqnarray}
Alternative estimators can be used instead of $\htau_t$ as defined above.
For instance,  the median of $\{|r_i^t|\}_{i\in[m]}$, 
can be used to define the  
alternative estimator:
\begin{eqnarray}
\htau_t^2 = \frac{1}{\Phi^{-1}(3/4)}
|r^t|_{(m/2)}\, ,\label{eq:ThresholdChoice2}
\end{eqnarray}
where $|u|_{(\ell)}$ is the $\ell$-th 
largest magnitude among the entries of a vector $u$, 
and $\Phi^{-1}(3/4)\approx 0.6745$ denotes the median of the absolute values of 
a Gaussian random variable.

By Proposition \ref{propo:Easy}, 
if the iteration converges to $(\hx,\hr)$,
then this is  a minimum of the LASSO cost function,
with regularization  parameter
\begin{eqnarray}
\lambda = \alpha\,  \frac{\|\hr\|_2}{\sqrt{m}}
\left(1-\frac{\|\hx\|_0}{m}\right)\, 
\end{eqnarray}
(in case the threshold is chosen as per
Eq.~(\ref{eq:ThresholdChoice1})).
While the relation between $\alpha$ and $\lambda$
is not fully explicit (it requires to find the 
optimum $\hx$), in practice $\alpha$ is as useful as  $\lambda$:
both play the role  of knobs that adjust the level of sparsity
of the  seeked solution.

We conclude by noting that the AMP algorithm (\ref{eq:dmm1}),
(\ref{eq:dmm2}) is quite close to iterative soft thresholding (IST), 
a well known algorithm for the same problem that proceeds by
\begin{eqnarray}
x^{t+1}&=&\eta(x^t + A^Tr^t\, ;\theta_t)\, ,\label{eq:ist1}\\
r^t &= &y - Ax^t\, .\label{eq:ist2}
\end{eqnarray}
The only (but important) difference lies in the introduction of
the term $\sb_tr^{t-1}$ in the second equation, cf. Eq.~(\ref{eq:dmm2}). 
This can be regarded as a momentum term
with a very specific prescription on its size, cf. Eq.~(\ref{eq:CoeffB}).
A similar term --with motivations analogous to the one 
presented below-- is popular under the name of `Onsager term' in
statistical physics \cite{Onsager,TAP,SpinGlass}.
%
%
\subsection{\dots and its derivation}
\label{sec:Derivation}

In this section we present an heuristic derivation of the AMP
iteration in Eqs.~(\ref{eq:dmm1}), (\ref{eq:dmm2}) starting from the 
standard message passing formulation given by
Eq.~(\ref{eq:MessPass1}) to (\ref{eq:MessPass4}).
Our objective is to develop an intuitive
understanding of the AMP iteration, as well as of the prescription
(\ref{eq:CoeffB}). Throughout our argument, we treat 
$m$ as scaling linearly with $n$. 
A full justification of the derivation 
presented here is beyond the scope of this
review: the actual rigorous analysis of the AMP algorithm goes through
an indirect and  very technical mathematical proof 
\cite{BM-MPCS-2010}.

We start by noticing that the sums
$\sum_{j\in\da\setminus i}A_{aj}^2\gamma^t_{j\to a}$ and 
$\sum_{b\in\di\setminus a}A^2_{bi}\beta^t_{b\to i}$ are sums of $\Theta(n)$
terms, each of order $1/n$
(because $A_{ai}^2=O(1/n)$). Notice that the terms in these sums 
are not independent: nevertheless by analogy to what 
happens in the case of sparse graphs 
\cite{MoT06,MontanariSparse,RiU08,AldousSteele}, one can hope that dependencies are weak. It is then reasonable to 
think that a law of large numbers applies and that therefore
these sums can be replaced by quantities that do not depend on the
instance or on the row/column index.

We then let $r^t_{a\to i} = \alpha^t_{a\to i}/\beta^t_{a\to i}$ and
rewrite the message passing iteration, cf. 
Eqs.~(\ref{eq:MessPass1}) to (\ref{eq:MessPass1}), as
\begin{align}
r_{a\to i}^t &= y_a - \sum_{j\in[n]\bs i}A_{aj}x_{j\to a}^t\, ,\label{eq:mp-repeated}\\
x_{i\to a}^{t+1}&=\eta\Big(\sum_{b\in[m]\bs a}A_{bi}r_{b\to i}^t;\theta_t
\Big)\, ,\label{eq:mp-repeated-bis}
\end{align}
where $\theta_t\approx
\lambda/\sum_{b\in\di\setminus a}A^2_{bi}\beta^t_{b\to i}$
is --as mentioned-- treated as independent of $b$.

Notice that on the right-hand side of both equations 
above, the messages appear
in sums over $\Theta(n)$ terms. Consider for instance the  messages
$\{r_{a\to i}^t\}_{i\in [n]}$ for a fixed node $a\in [m]$.
These depend on $i\in [n]$ only because the term  excluded
from the sum on the right hand side of Eq.~(\ref{eq:mp-repeated})
changes. It is therefore natural to guess that
$r^{t}_{a\to i}=r^t_a+O(n^{-1/2})$ and
$x^{t}_{i\to a}=x^t_i+O(m^{-1/2})$, where
$r^t_a$ only depends on the index $a$ (and not on $i$),
and $x^t_i$ only depends on $i$ (and not on $a$).

A naive approximation would consist in neglecting
the $O(n^{-1/2})$ correction
but this approximation 
turns out to be inaccurate even in the
large-$n$ limit. We instead set
\begin{eqnarray*}
r_{a\to i}^t = r_a^t+\dr_{a\to i}^t\, ,\;\;\;\;\;\;\;
x_{i\to a}^t = x_i^t+\dx_{i\to a}^t\, .
\end{eqnarray*}
Substituting in Eqs.~(\ref{eq:mp-repeated}) and 
(\ref{eq:mp-repeated-bis}), we get
\begin{align*}
r_{a}^t+\dr_{a\to i}^t &= y_a - \sum_{j\in[n]}A_{aj}
(x_{j}^t+\dx_{j\to a}^t)
+A_{ai}(x_{i}^t+\dx_{i\to a}^t)\, ,\\
x_{i}^{t+1}+\dx_{i\to a}^{t+1}&=
\eta\Big(\sum_{b\in[m]}A_{bi}(r_{b}^t+\dr_{b\to i}^t)-
A_{ai}(r_{a}^t+\dr_{a\to i}^t);\,\theta_t\Big)\, .
\end{align*}
We will now drop the terms that are negligible without writing
explicitly the error terms. First of all notice that single terms of
the type $A_{ai}\dr_{a\to i}^t$ are of order $1/n$ and
can be safely neglected. Indeed $\dr_{a\to i} = O(n^{-1/2})$
by our ansatz, and $A_{ai} = O(n^{-1/2})$ by definition.
We get
\begin{align*}
r_{a}^t+\dr_{a\to i}^t &= y_a - \sum_{j\in[n]}A_{aj}
(x_{j}^t+\dx_{j\to a}^t)
+A_{ai}x_{i}^t\, ,\\
x_{i}^{t+1}+\dx_{i\to a}^{t+1}&=
\eta\Big(\sum_{b\in[m]}A_{bi}(r_{b}^t+\dr_{b\to i}^t)-
A_{ai}r_{a}^t;\theta_t\Big)\, .
\end{align*}
We next expand the second equation to linear order in $\dx_{i\to a}^t$ and
$\dr_{a\to i}^t$:
\begin{align*}
z_{a}^t+\dr_{a\to i}^t &= y_a - \sum_{j\in[n]}A_{aj}
(x_{j}^t+\dx_{j\to a}^t)
+A_{ai}x_{i}^t\, ,\\
x_{i}^{t+1}+\dx_{i\to a}^{t+1}&=
\eta\Big(\sum_{b\in[m]}A_{bi}(r_{b}^t+\dr_{b\to i}^t);\theta_t\Big)-
\eta'\Big(\sum_{b\in[m]}A_{bi}(r_{b}^t+\dr_{b\to i}^t);\theta_t\Big)
A_{ai}z_{a}^t\, .
\end{align*}
The careful reader might be puzzled by the fact that 
the soft thresholding function $u\mapsto \eta(u;\theta)$
is non-differentiable at $u\in \{+\theta,-\theta\}$.
However, the rigorous analysis carried out in \cite{BM-MPCS-2010}
through a different (and more technical) methods
reveals that almost-everywhere differentiability is sufficient 
here. 

Notice that the last term on the right hand side of
the first equation above is the only one dependent on $i$,
and we can therefore identify this term with $\dr_{a\to i}^t$.
We obtain the decomposition
\begin{align}
r_{a}^t &= y_a - \sum_{j\in[n]}A_{aj}
(x_{j}^t+\dx_{j\to a}^t)\, ,\label{eq:Z1}\\
\dr_{a\to i}^t &= A_{ai}x_{i}^t\, .\label{eq:Z2}
\end{align}
Analogously for the second equation we get
\begin{align}
x_{i}^{t+1}&=
\eta\Big(\sum_{b\in[m]}A_{bi}(r_{b}^t+\dr_{b\to i}^t);\theta_t\Big)\, ,
\label{eq:X1}\\
\dx_{i\to a}^{t+1} & =-
\eta'\Big(\sum_{b\in[m]}A_{bi}(r_{b}^t+\dr_{b\to i}^t);\theta_t\Big)
A_{ai}r_{a}^t\, .\label{eq:X2}
\end{align}

Substituting Eq.~(\ref{eq:Z2}) in Eq.~(\ref{eq:X1}) to eliminate
$\dr_{b\to i}^t$ we get
\begin{align}
x_{i}^{t+1}&=
\eta\Big(\sum_{b\in[m]}A_{bi}r_{b}^t+\sum_{b\in[m]}A_{bi}^2x_{i}^t;
\theta_t\Big)\, ,
\end{align}
and using the normalization of $A$, we get
$\sum_{b\in [m]}A_{bi}^2\to 1$, whence
\begin{align}
x^{t+1}&=
\eta(x^t+A^Tr^t;\theta_t)\, .
\end{align}

Analogously substituting Eq.~(\ref{eq:X2}) in (\ref{eq:Z1}),
we get
\begin{align}
z_{a}^t &= y_a - \sum_{j\in[n]}A_{aj}x_{j}^t+
\sum_{j\in[n]}A_{aj}^2\eta'(x^{t-1}_j+(A^Tr^{t-1})_j;\theta_{t-1})r_{a}^{t-1}\, .\label{eq:ZZZZ}
\end{align}
Again, using the law of large numbers  and the normalization of $A$,
we get
\begin{align}
\sum_{j\in[n]}A_{aj}^2\eta'(x^{t-1}_j+(A^Tr^{t-1})_j;\theta_{t-1})
\approx \frac{1}{m}\sum_{j\in[n]}\eta'(x^{t-1}_j+(A^Tr^{t-1})_j;\theta_{t-1})
=\frac{1}{m}\|x^t\|_0\, ,
\end{align}
whence substituting in (\ref{eq:ZZZZ}),
we obtain Eq.~(\ref{eq:dmm2}), with the prescription 
(\ref{eq:CoeffB}) for the Onsager term.
This finishes our derivation.

%
%
\section{High-dimensional analysis}
\label{sec:LargeSystem}

The AMP algorithm enjoys several unique properties.
In particular it admits an \emph{asymptotically exact}
analysis along sequences of instances of diverging size. 
This is quite remarkable, since all analysis available for other 
algorithms that solve the LASSO hold only `up to undetermined
constants'.

In particular in the large system limit (and with 
the exception of a `phase transition' line), AMP can be shown to converge 
exponentially fast to the LASSO optimum. Hence the analysis 
of AMP yields asymptotically exact predictions
on the behavior of the LASSO, including in particular the asymptotic
mean square error per variable.

%
%
\subsection{Some numerical experiments with AMP}
\label{sec:NumExp}

\begin{figure}
\includegraphics[width=8cm,angle=0]{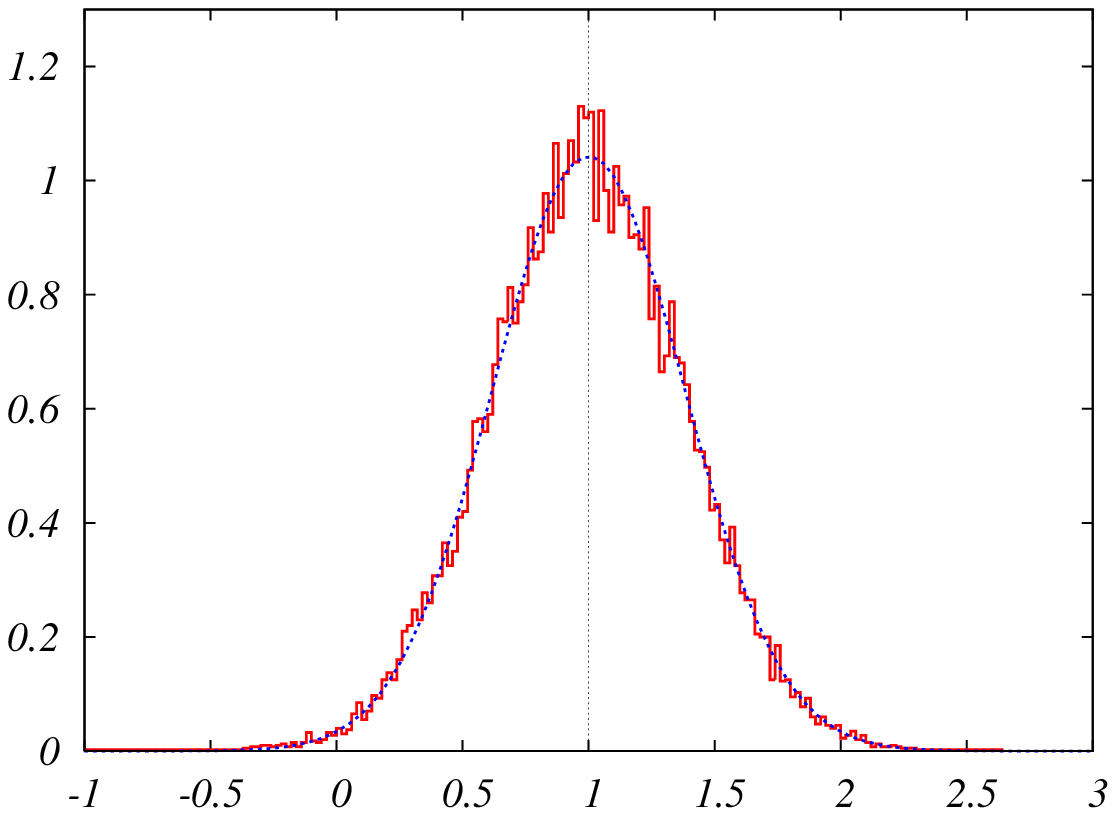}
\includegraphics[width=8cm,angle=0]{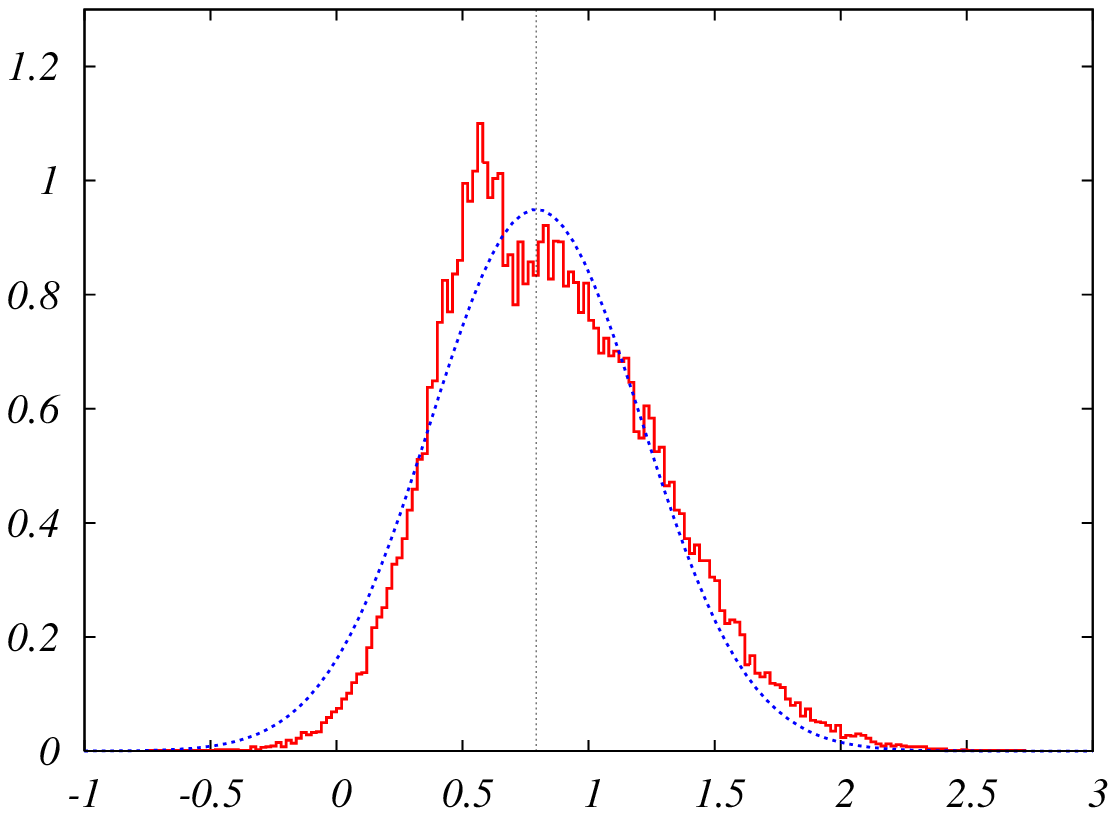}
\put(-140,-8){$(x^t+A^Tr^t)_i$}
\put(-360,-8){$(x^t+A^Tr^t)_i$}
\caption{{\small Distributions of un-thresholded estimates for
AMP (left) and IST (right), after $t=10$ iterations.
These data were obtained using
sensing matrices with $m=2000$, $n=4000$ and i.i.d. entries 
uniform in $\{+1/\sqrt{m},-1/\sqrt{m}\}$. The signal $x$
contained $500$ non-zero entries uniform in $\{+1,-1\}$. 
A total of $40$ instances was used to build the histograms.
Blue lines are Gaussian fits and vertical lines represent the fitted mean.}}
\label{fig:GaussHisto}
\end{figure}
How is it possible that an \emph{asymptotically exact}
analysis of AMP can be carried out? Figure \ref{fig:GaussHisto}
illustrates the key point. It shows the distribution of 
un-thresholded estimates $(x^t+A^Tr^t)_i$ for coordinates $i$
such that the original signal had value $x_i = +1$.
These estimates were obtained using the AMP algorithm 
(\ref{eq:dmm1}), (\ref{eq:dmm2}) with choice
(\ref{eq:CoeffB}) of $\sb_t$ (plot on the left) and
the iterative soft thresholding algorithm (\ref{eq:ist1}), (\ref{eq:ist2})
(plot on the right). The 
same instances (i.e. the same matrices $A$ and measurement vectors $y$)
were used in the two cases, but the resulting distributions
are dramatically different. In the case of AMP, the distribution 
is close to Gaussian, with mean on the correct value, $x_i=+1$.
For iterative soft thresholding the estimates do not have the correct mean
and are not Gaussian.

\begin{figure}
\includegraphics[width=3.2in]{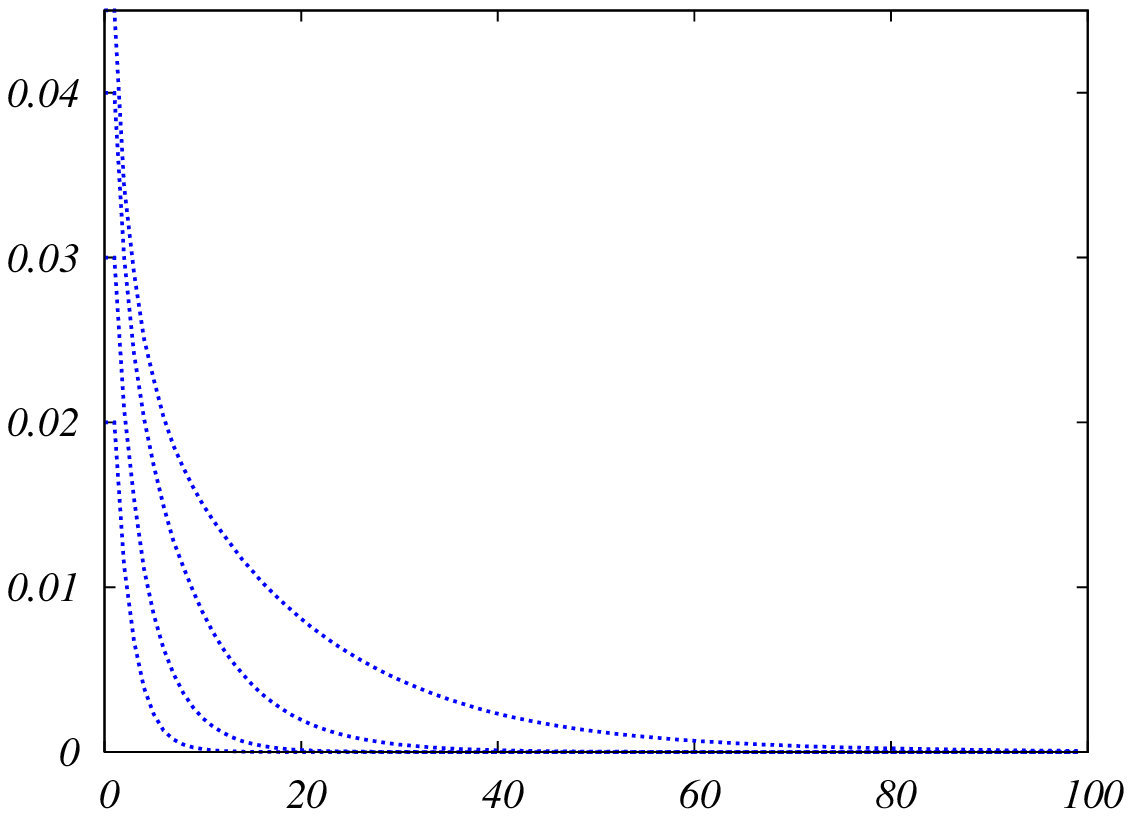}
\includegraphics[width=3.2in]{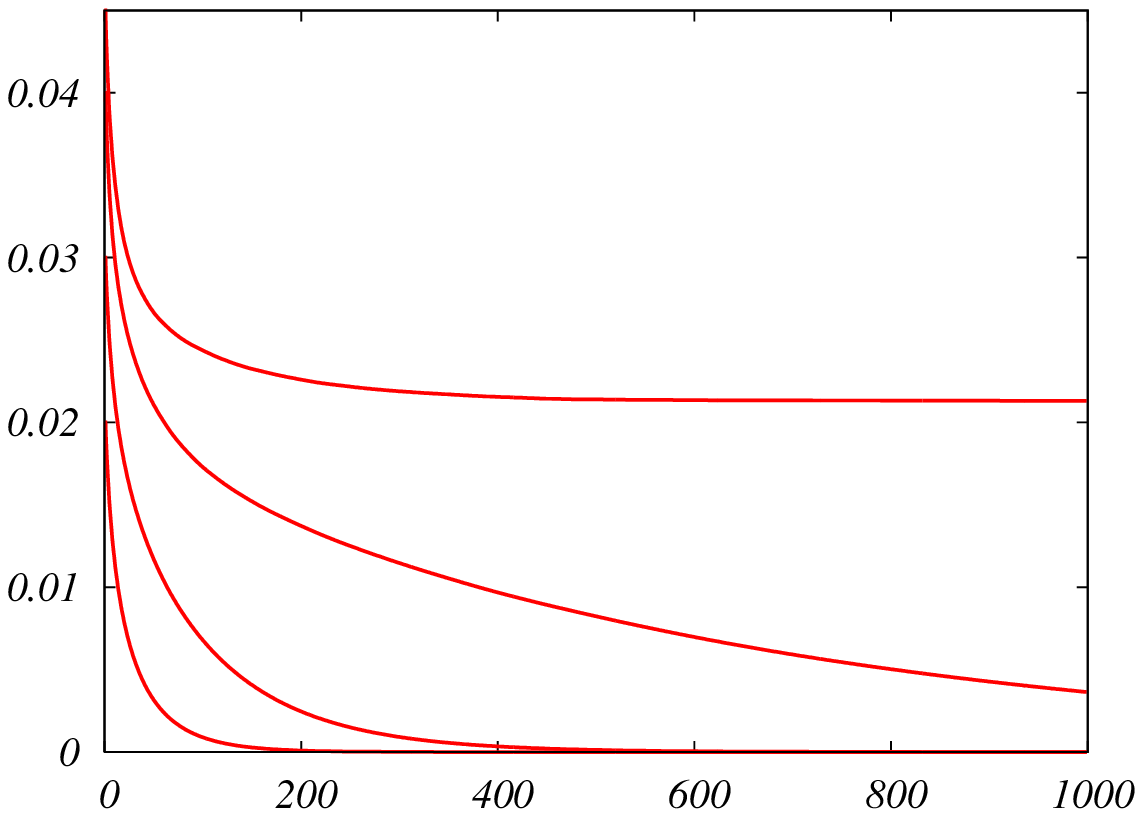}
\put(-360,-5){iterations}
\put(-130,-5){iterations}
\put(-480,90){MSE}
\put(-235,90){MSE}
  \caption{Evolution of the mean square error as a function of the number
of iterations for AMP (left) and iterative soft thresholding (right),
for random measurement matrices $A$, with i.i.d. entries
$A_{ai}\in\{+1/\sqrt{m},-1/\sqrt{m}\}$ uniformly.
Notice the different scales used for the horizontal axis!
Here $n=8000$, $m=1600$. Different curves depends to
different levels of sparsity. The number of
non-zero entries of the signal $x$ is, for the various curves,  
$\|x\|_0 = 800$, $1200$, $1600$, $1800$ (from bottom to top).}
\label{fig:ConvergenceExp}
\end{figure}

This phenomenon appears here as an empirical observation, valid
for a specific iteration number $t$, and specific dimensions $m,n$. 
In the next section 
we will explain that it  can be proved rigorously in the limit of a 
large number of dimensions, for all values of iteration number $t$. 
Namely,  as $m,n\to\infty$ at $t$ fixed, 
the empirical distribution of  $\{(x^t+A^Tr^t)_i-x_i\}_{i\in[n]}$
converges to a gaussian distribution,
when $x^t$ and $r^t$ are computed using AMP.
The  variance of this distribution 
depends on $t$, and the its evolution with $t$
can be computed exactly.
 Viceversa,
for iterative soft thresholding, the distribution of the 
same quantities remains non-gaussian.

This dramatic difference remains true for any $t$, even when 
AMP and IST converge the same minimum.
Indeed even at the fixed point,  the resulting residual $r^t$ 
is different in the two algorithms, as a consequence of the 
introduction Onsager term.

More importantly, the two algorithms differ 
dramatically in the rate of convergence.
One can interpret the vector $(x^t+A^Tr^t)-x$ as `effective noise'
after $t$ iterations. Both AMP and IST
 `denoise' the vector $(x^t+A^Tr^t)$ using the soft thresholding operator.
As discussed in Section \ref{sec:Scalar}, 
the soft thresholding operator is essentially optimal for denoising in
gaussian noise. This suggests that AMP should have superior 
performances (in the sense of faster
convergence to the LASSO minimum)
with respect to simple IST.

Figure \ref{fig:ConvergenceExp} presents the results 
of a small experiment confirming this expectation.
Measurement matrices $A$ with dimensions $m=1600$, $n=8000$,
were generated randomly with i.i.d. entries
$A_{ai}\in\{+1/\sqrt{m},-1/\sqrt{m}\}$ uniformly at random.
We consider here the problem of reconstructing a signal $x$
with entries $x_i\in\{+1,0,-1\}$ from noiseless measurements $y = Ax$, 
for different levels of sparsity.
Thresholds were set according to the prescription
(\ref{eq:ThresholdChoice1}) with $\alpha=1.41$ for AMP
(the asymptotic theory of \cite{DMM09} yields the prescription
$\alpha \approx 1.40814$) and $\alpha = 1.8$ for IST 
(optimized empirically). For the latter
algorithm, the matrix $A$ was rescaled in order to get an
operator norm $\|A\|_2 = 0.95$.

Convergence to the original signal $x$ is slower and slower 
as this becomes less and less sparse\footnote{Indeed 
basis pursuit (i.e. reconstruction via $\ell_1$ minimization)
fails  with high probability if $\|x\|_0/m\gtrsim 0.243574$,
see \cite{DonohoCentrally} and Section \ref{sec:NSPT}.}.
Overall, AMP appears to be at least 10 times faster even 
on the sparsest vectors (lowest curves in the figure).
%
%
\subsection{State evolution}
\label{sec:StateEvolution}

State evolution describes the asymptotic limit
of the AMP estimates as $m,n\to \infty$, for any fixed $t$.
The word `evolution' refers to the fact that one
obtains an `effective' evolution with $t$. The word
`state' refers to the fact that the algorithm behavior is captured in this 
limit by a single parameter (a state) $\tau_t\in\reals$. 

We will consider sequences of instances
of increasing sizes, along which the AMP algorithm behavior 
admits a non-trivial limit. 
An instance is completely determined by the measurement matrix $A$, 
the signal $x$, and the noise vector $w$, the vector
of measurements $y$ being  given by $y=Ax+w$,
cf. Eq.~(\ref{eq:FirstModel}).
While rigorous results have been proved
so far only in the case in which 
the sensing matrices $A$ have i.i.d. Gaussian entries, 
it is nevertheless useful to collect a few basic properties that
the sequence needs to satisfy in order for state evolution to hold.
\begin{definition}\label{def:Converging}
The sequence of instances $\{x(n), w(n), A(n)\}_{n\in\naturals}$
indexed by $n$ is said to be a \emph{converging sequence}
if
$x(n)\in\reals^{n}$, $w(n)\in\reals^m$, $A(n)\in\reals^{m\times n}$
with $m=m(n)$ is  such that $m/n\to\delta\in(0,\infty)$,
and in addition the following conditions hold:
\begin{itemize}
\item[$(a)$] The empirical distribution of the entries of $x(n)$
converges weakly to a probability measure $p_{0}$ on $\reals$
with bounded second moment. Further
$n^{-1}\sum_{i=1}^nx_{i}(n)^2\to \E_{p_{0}}\{X_0^{2}\}$.
\item[$(b)$] The empirical distribution of the entries of $w(n)$
converges weakly to a probability measure $p_{W}$ on $\reals$
with bounded second moment. Further
$m^{-1}\sum_{i=1}^mw_{i}(n)^2\to \E_{p_{W}}\{W^{2}\}\equiv\sigma^2$.
\item[$(c)$]If $\{e_i\}_{1\le i\le n}$, $e_i\in\reals^n$ denotes the canonical
basis, then $\lim_{n\to \infty}\max_{i\in [n]}\|A(n)e_i\|_2=1$,\\
$\lim_{n\to \infty}\min_{i\in [n]}\|A(n)e_i\|_2= 1$.
\end{itemize}
\end{definition}
As mentioned above, rigorous results have been proved only 
for a subclass of converging sequences, namely under the assumption
that the matrices $A(n)$ have i.i.d. Gaussian entries.
Notice that such matrices satisfy condition $(c)$ by elementary tail 
bounds on $\chi$-square random variables. The same condition
is satisfied by matrices with i.i.d. subgaussian entries thanks to
concentration inequalities \cite{Ledoux}.

On the other hand, numerical simulations show that the same limit 
behavior should apply within a much broader domain, including for
instance random matrices with i.i.d. entries under an appropriate
moment condition. This 
\emph{universality} phenomenon is well-known in random matrix theory
whereby asymptotic results initially established for Gaussian matrices
were subsequently proved for  broader classes of matrices.
Rigorous evidence in this direction is presented in \cite{KM-2010}.
This paper shows that
the normalized cost $\min_{x\in\reals^{n}}\cost_{A(n),y(n)}(x)/n$ 
has a limit for $n\to\infty$,  which is universal with respect to
random matrices $A$ with i.i.d. entries.
(More precisely, it is universal provided $\E\{A_{ij}\}=0$,
$\E\{A_{ij}^2\}=1/m$ and $\E\{A_{ij}^6\}\le C/m^{3}$ for some $n$-independent
constant $C$.)

For a converging sequence of instances  $\{x(n), w(n), A(n)\}_{n\in\naturals}$,
and an arbitrary sequence of thresholds $\{\theta_t\}_{t\ge 0}$
(independent of $n$), the AMP iteration
(\ref{eq:dmm1}), (\ref{eq:dmm2}) admits a high-dimensional limit 
which can be characterized exactly, provided Eq.~(\ref{eq:CoeffB}) 
is used for fixing the Onsager term. This limit is given in terms 
of the trajectory of a simple one-dimensional iteration termed
\emph{state evolution} which we will describe next.

Define the sequence $\{\tau_t^2\}_{t\ge 0}$
by setting  $\tau_{0}^2 =\sigma^2+\E\{X_0^2\}/\delta$
(for $X_0\sim p_{0}$ and $\sigma^2\equiv \E\{W^2\}$, $W\sim p_W$)
and letting, for all $t\ge 0$:
\begin{eqnarray}
\tau_{t+1}^2 & = & \seF(\tau_t^2,\theta_t)\, ,\label{eq:1-dim-SE}\\
\seF(\tau^2,\theta) &\equiv &\sigma^2+\frac{1}{\delta}\,
\E\{\,[\eta(X_0+\tau Z;\theta)-X_0]^2\}\,,
\end{eqnarray}
where $Z\sim\normal(0,1)$ is independent of $X_0\sim p_0$. Notice that the function $\seF$ depends implicitly on the law $p_{0}$.
Further, the state evolution $\{\tau_t^2\}_{t\ge 0}$  depends
on the specific converging sequence through the law
$p_0$, and the second moment of the noise  
$\E_{p_W}\{W^2\}$, cf. Definition \ref{def:Converging}.

We say a function $\psi:\reals^k\to\reals$ is \emph{pseudo-Lipschitz} if there exist a constant $L>0$ such that for all $x,y\in\reals^k$: $|\psi(x)-\psi(y)|\le L(1+\|x\|_2+\|y\|_2)\|x-y\|_2$.
(This is a special case of the definition used
in \cite{BM-MPCS-2010} where such a function is called pseudo-Lipschitz
\emph{of order 2}.)

The following theorem was conjectured in \cite{DMM09}, and proved in
\cite{BM-MPCS-2010}. It shows that the behavior of $\AMP$ can be tracked by the
above state evolution  recursion.
\begin{theorem}[\cite{BM-MPCS-2010}]\label{prop:state-evolution} Let $\{x(n), w(n), A(n)\}_{n\in\naturals}$ be a converging sequence of
instances with the entries of $A(n)$ i.i.d. normal with mean $0$ and variance
$1/m$, while the signals $x(n)$ and noise
vectors $w(n)$ satisfy the hypotheses of Definition
\ref{def:Converging}. Let $\psi_1:\reals\to \reals$,
$\psi_2:\reals\times\reals\to \reals$ be  pseudo-Lipschitz functions. 
Finally, let $\{x^t\}_{t\ge 0}$, $\{r^t\}_{t\ge 0}$ be the sequence of 
estimates and residuals produced by AMP, cf. Eqs.~(\ref{eq:dmm1}), 
(\ref{eq:dmm2}).
Then, almost surely
\begin{eqnarray}
\lim_{n\to\infty}\frac{1}{m}\sum_{a=1}^m\psi_1
\big(r_{a}^{t}\big) & = &\E\Big\{\psi_1\big(\tau_t Z\big)\Big\}\, ,\label{eq:state-evolution-1}\\
\lim_{n\to\infty}\frac{1}{n}\sum_{i=1}^n\psi_2
\big(x_{i}^{t+1},x_{i}\big) &= &\E\Big\{\psi_2\big(\eta(X_0+\tau_t Z;\theta_t),X_0\big)\Big\}\, ,\label{eq:state-evolution-2}
\end{eqnarray}
where $Z\sim\normal(0,1)$ is independent of $X_0\sim p_{0}$.
\end{theorem}
It is worth pausing for a few remarks.
\begin{remark}
\emph{
Theorem \ref{prop:state-evolution} holds for any choice of the sequence of 
thresholds $\{\theta_t\}_{t\ge 0}$. It 
does not require --for instance-- that the latter
converge. Indeed \cite{BM-MPCS-2010} proves a more general result
that holds for a a broad class of approximate message passing 
algorithms. The more general theorem establishes the validity 
of} state evolution \emph{in this broad context.}
 
\emph{For instance, the soft thresholding functions $\eta(\,\cdot\,;\theta_t)$ 
can be replaced by a generic sequence of Lipschitz continuous functions,
provided the coefficients $\sb_t$
in Eq.~(\ref{eq:dmm2}) are suitably modified.}
\end{remark}

\begin{remark}
\emph{This theorem does not require the vectors $x(n)$ to be
sparse. The use of other functions instead of the soft thresholding functions 
$\eta(\,\cdot\,;\theta_t)$ in the algorithm
can be useful for estimating such non-sparse vectors.}

\emph{Alternative nonlinearities, 
can also be useful when additional information on the entries of 
$x(n)$ is available.}
\end{remark}

\begin{remark}
\emph{While the theorem requires the matrices $A(n)$ to be random,
neither the signal $x(n)$ nor the noise vectors $w(n)$ need to be random.
They are generic deterministic sequences of vectors under the
conditions of Definition \ref{def:Converging}.}

\emph{The fundamental reason for this universality is that the matrix $A$
is both row and column exchangeable. Row exchangeability
guarantees universality with respect to the signals $x(n)$,
while column exchangeability guarantees universality with respect to
the noise $w(n)$. To see why, observe that, by row exchangeability 
(for instance), $x(n)$ can be replaced by the random vector obtained by 
randomly permuting its entries. Now, the distribution of such a random 
vector is very close (in appropriate sense) to the one of a random 
vector with i.i.d. entries whose distribution matches the empirical 
distribution of $x(n)$.}
\end{remark}

Theorem \ref{prop:state-evolution} strongly supports
both the use of soft thresholding, and the choice of the threshold 
level in Eq.~(\ref{eq:ThresholdChoice1}) or  (\ref{eq:ThresholdChoice2}).
Indeed Eq.~(\ref{eq:state-evolution-1}) states that the components of $r^t$
are approximately i.i.d.  $\normal(0,\tau_t^2)$, and hence both 
definitions of $\htau_t$ in 
Eq.~(\ref{eq:ThresholdChoice1}) or  (\ref{eq:ThresholdChoice2})
provide consistent estimators of $\tau_t$. Further,
Eq.~(\ref{eq:state-evolution-1}) implies that the
components of the deviation $(x^t+A^Tr^t-x)$ are also approximately 
i.i.d.  $\normal(0,\tau_t^2)$. In other words, the estimate
$(x^t+A^Tr^t)$ is equal to the actual signal plus 
noise of variance $\tau_t^2$, as illustrated in Fig.~\ref{fig:GaussHisto}.
According to our discussion of scalar estimation
in Section \ref{sec:Scalar}, the correct way of reducing the noise 
is to apply soft thresholding with threshold level $\alpha\tau_t$.

\begin{figure}
\centering
  \includegraphics[width=3.5in]{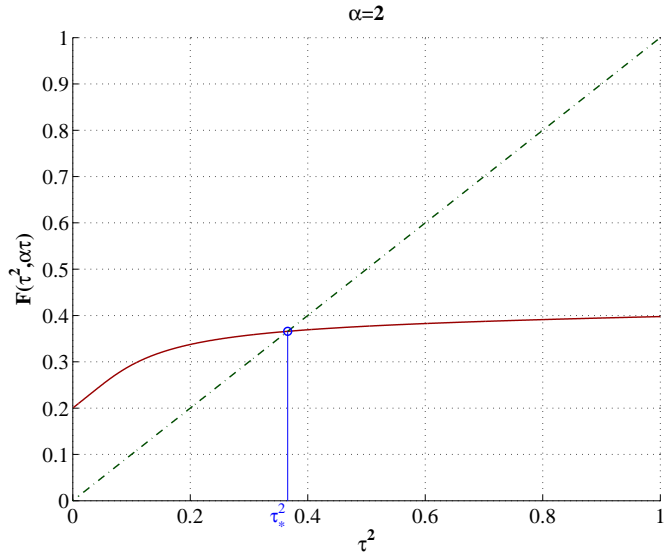}
  \caption{Mapping $\tau^2\mapsto\seF(\tau^2,\alpha\tau)$ for $\alpha=2$, $\delta=0.64$, $\sigma^2=0.2$,  $p_{0}(\{+1\}) = p_{0}(\{-1\}) = 0.064$ and
$p_{0}(\{0\}) = 0.872$. }
  \label{fig:tau2->F(tau2)}
\end{figure}
The choice $\theta_t = \alpha\tau_t$ with $\alpha$ fixed
has another important advantage.
In this case, the sequence $\{\tau_t\}_{t\ge 0}$ is
determined by the one-dimensional  recursion
\begin{eqnarray}
\tau_{t+1}^2 = \seF(\tau_t^2,\alpha\tau_t)\, .
\end{eqnarray}
The function $\tau^2\mapsto\seF(\tau^2,\alpha\tau)$ 
depends on the distribution of $X_0$ as well as on the other
parameters of the problem. An example is plotted in
Fig.~(\ref{fig:tau2->F(tau2)}). It turns out that the behavior 
shown here is generic: the function is always non-decreasing and
concave. This remark allows to easily prove the following.
\begin{proposition}[\cite{NSPT}]\label{propo:UniqFP}
Let $\alpha_{\rm min}= \alpha_{\rm min}(\delta)$ be the unique
non-negative solution of the equation 
\begin{eqnarray}
(1+\alpha^2)\Phi(-\alpha)-\alpha\phi(\alpha) = \frac{\delta}{2}\, ,
\label{eq:AlphaMin}
\end{eqnarray}
with $\phi(z) \equiv e^{-z^2/2}/\sqrt{2\pi}$ the standard Gaussian density
and $\Phi(z) \equiv\int_{-\infty}^{z}\phi(x)\,\de x$.

For any $\sigma^2>0$, $\alpha>\alpha_{\rm min}(\delta)$,
the  fixed point equation
$\tau^2 = \seF(\tau^2,\alpha\tau)$ admits a unique solution.
Denoting by $\tau_*=\tau_*(\alpha)$ this solution, we
have $\lim_{t\to\infty}\tau_t=\tau_*(\alpha)$.  
\end{proposition}
It can also be shown that, under the choice 
$\theta_t=\alpha\tau_t$, convergence is exponentially fast
unless the problem parameters take some `exceptional' values
(namely on the phase transition boundary discussed below).
%
%
\subsection{The risk of the LASSO}\label{sec:Results}

State evolution provides a scaling limit of the AMP 
dynamics in the high-dimensional setting. By showing that AMP
converges to the LASSO estimator, one can transfer this information
to a scaling limit result of the LASSO estimator itself.

Before stating the limit,
we have to describe a \emph{calibration} mapping between 
the AMP parameter $\alpha$ (that defines the 
sequence of thresholds $\{\theta_t\}_{t\ge 0}$) and the LASSO regularization 
parameter $\lambda$. The connection was first introduced in \cite{NSPT}.

We define the function $\alpha\mapsto \lambda(\alpha)$
on $(\alpha_{\rm min}(\delta),\infty)$, by
\begin{eqnarray}
\lambda(\alpha) \equiv \alpha\tauinf\left[1 - \frac{1}{\delta}
\prob\big\{|X_0+\tauinf Z|\ge \alpha\tauinf\big\}\right]\, ,
\label{eq:calibration}
\end{eqnarray}
where $\tauinf=\tauinf(\alpha)$ is the state evolution fixed point 
defined as per Proposition \ref{propo:UniqFP}.
Notice that this relation corresponds to the scaling limit
of the general relation (\ref{eq:GeneralLambdaTheta}),
provided we assume that the solution of the LASSO optimization
problem (\ref{eq:LASSO}) 
is indeed described by the fixed point of state evolution
(equivalently, by its $t\to\infty$ limit). 
This follows by noting that $\theta_t\to\alpha\tau_*$
and that $\|x\|_0/n\to \E\{\eta'(X_0+\tauinf Z;\alpha\tauinf)\}$.
While this is just an interpretation of
the definition (\ref{eq:calibration}),
the result presented next implies that the interpretation is
indeed correct.

In the following we will need to invert the function $\alpha\mapsto
\lambda(\alpha)$.
We thus define $\alpha:(0,\infty)\to(\alpha_{\rm min},\infty)$
in such a way that
\begin{eqnarray*}
\alpha(\lambda) \in \big\{\, a\in (\alpha_{\rm min},\infty)\, :\,
\lambda(a) =\lambda\big\}\, .\label{eq:AlphaOfLambda}
\end{eqnarray*}
The fact that the right-hand side
is non-empty, and therefore the function $\lambda\mapsto\alpha(\lambda)$
is well defined, is part of the main result of this section.
\begin{theorem}\label{thm:Risk}
Let $\{x(n), w(n), A(n)\}_{n\in\naturals}$ be a converging sequence of
instances with the entries of $A(n)$ i.i.d. normal with mean $0$ and variance
$1/m$. Denote by $\hx(\lambda)$ the \LASSO\, estimator for
instance $(x(n), w(n), A(n))$, with $\sigma^2,\lambda> 0$, and
let $\psi:\reals\times\reals\to \reals$ be a pseudo-Lipschitz function.
Then, almost surely
\begin{eqnarray}\label{eq:asymptotic-result}
\lim_{n\to\infty}\frac{1}{n}\sum_{i=1}^n\psi
\big(\hx_{i},x_{i}\big) = \E\Big\{\psi\big(\eta(X_0+\tau_* Z;\theta_*),X_0\big)\Big\}\, ,
\end{eqnarray}
where $Z\sim\normal(0,1)$ is independent of $X_0\sim p_{0}$,
$\tau_*=\tau_*(\alpha(\lambda))$ and $\theta_*=\alpha(\lambda)
\tau_*(\alpha(\lambda))$.

Further, the function $\lambda\mapsto\alpha(\lambda)$ 
is well defined and unique on $(0,\infty)$.
\end{theorem}
The assumption of a converging problem sequence is
important for the result to hold, while the hypothesis of Gaussian
measurement matrices $A(n)$ is necessary for the proof technique
to be applicable. On the other hand,
the restrictions $\lambda,\sigma^2>0$, and $\prob\{X_0\neq 0\}>0$
(whence $\tau_*\neq0$ using Eq.~\eqref{eq:calibration})
are made in order to avoid technical complications due
to degenerate cases. Such cases can be resolved by continuity arguments.

Let us emphasize that some of the remarks made in the
case of state evolution, cf. Theorem \ref{prop:state-evolution}, hold 
for the last theorem as well. More precisely.
\begin{remark}
\emph{Theorem \ref{thm:Risk} does not require
either the signal $x(n)$ or the noise vectors $w(n)$ to be random.
They are generic deterministic sequences of vectors under the
conditions of Definition \ref{def:Converging}.}

\emph{In particular, it does not require the vectors $x(n)$ to be
sparse. Lack of sparsity will reflect in a large risk
as computed through the mean square error computed through
Eq.~(\ref{eq:asymptotic-result}).} 

\emph{On the other hand, when restricting $x(n)$ to be $k$
sparse for $k = n\ve$ (i.e. to be in the class $\cF_{n,k}$), 
one can derive asymptotically exact estimates for 
the minimax risk over this class. This will be further 
discussed in Section \ref{sec:NSPT}.} 
\end{remark}
\begin{remark}
\emph{As a special case, for noiseless measurements $\sigma=0$,
and as $\lambda\to 0$, the above formulae describe the 
asymptotic risk (e.g. mean square error) for the
basis pursuit estimator, minimize $\|x\|$ subject to 
$y=Ax$. For sparse signals  $x(n)\in \cF_{n,k}$, $k=n\rho\delta$,
the risk vanishes below a certain phase transition
line $\rho<\rho_{\rm c}(\delta)$: this point is further
discussed in Section \ref{sec:NSPT}.}
\end{remark}

Let us now discuss some limitations of this result.
Theorem \ref{thm:Risk} assumes that the entries of matrix $A$ are 
i.i.d. Gaussians.
Further, our result is asymptotic,
and one might wonder how accurate it is for
instances of moderate dimensions.

Numerical simulations were carried out in \cite{NSPT,OurLASSO_Exp} and
suggest that the result is universal over a broader class of 
matrices and that is relevant already for
$n$ of the order of a few hundreds.
As an illustration, we present in Figs.~\ref{fig:GaussianMatrices}
and \ref{fig:PM1Matrices} the outcome of such simulations
for two types of random matrices.
Simulations with real data can be found in \cite{OurLASSO_Exp}.
We generated the signal vector randomly with entries in $\{+1,0,-1\}$
and $\prob(x_{0,i}=+1) = \prob(x_{0,i}=-1) = 0.064$. The noise vector
$w$ was generated by using i.i.d. $\normal(0,0.2)$ entries.

We solved the LASSO problem (\ref{eq:LASSO}) and
computed estimator $\hx$ using \texttt{CVX}, a package for
specifying and solving convex programs \cite{CVX} and \texttt{OWLQN},
a package for solving large-scale versions of LASSO \cite{OWLQN}. We used
several values of $\lambda$ between $0$ and $2$ and $n$ equal to $200$,
$500$, $1000$, and $2000$. The aspect ratio of matrices was fixed in
all cases to $\delta=0.64$.
For each case, the point $(\lambda,\MSE)$ was plotted and the results
are shown in the figures. Continuous lines corresponds to the
asymptotic prediction by Theorem
\ref{thm:Risk} for $\psi(a,b) = (a-b)^2$, namely
\begin{eqnarray*}
\lim_{n\to\infty}\frac{1}{n}\|\hx-x\|_2^2 = \E\big\{
\big[\eta(X_0+\tau_*Z;\theta_*)-X_0\big]^2\big\} =
\delta(\tau_*^2-\sigma^2)\, .
\end{eqnarray*}
 The agreement is remarkably good already for $n,m$ of
the order of a few hundreds, and deviations are consistent with statistical
fluctuations.

The two figures correspond to different entries distributions:
$(i)$ Random Gaussian matrices with aspect ratio $\delta$ and i.i.d.
$\normal(0,1/m)$ entries (as in Theorem \ref{thm:Risk});
$(ii)$ Random $\pm1$ matrices with aspect ratio $\delta$. Each entry is
independently equal to $+1/\sqrt{m}$ or $-1/\sqrt{m}$ with equal probability.
The resulting MSE curves are hardly distinguishable.
Further evidence towards universality will be discussed in 
Section~\ref{sec:Universality}.

Notice that the asymptotic prediction has
a minimum as a function of $\lambda$. The location of this minimum
 can be used to select the regularization parameter.

\begin{figure}
\centering
  \includegraphics[width=4.in]{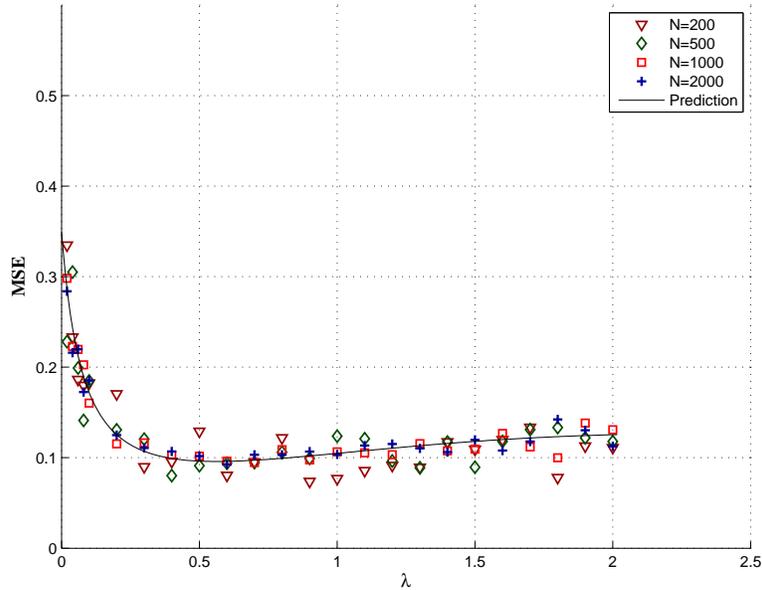}
  \caption{Mean square error (MSE) as a function of the regularization
parameter $\lambda$ compared to the asymptotic prediction for $\delta=0.64$ and $\sigma^2=0.2$. Here the measurement matrix $A$ has i.i.d.  $\normal(0,1/m)$ entries. Each point in this plot is generated by finding the LASSO predictor $\hx$ using a measurement vector $y=Ax+w$ for an independent signal vector $x$, an independent noise vector $w$, and an independent matrix $A$.}\label{fig:GaussianMatrices}
\end{figure}

\begin{figure}
\centering
  \includegraphics[width=4.in]{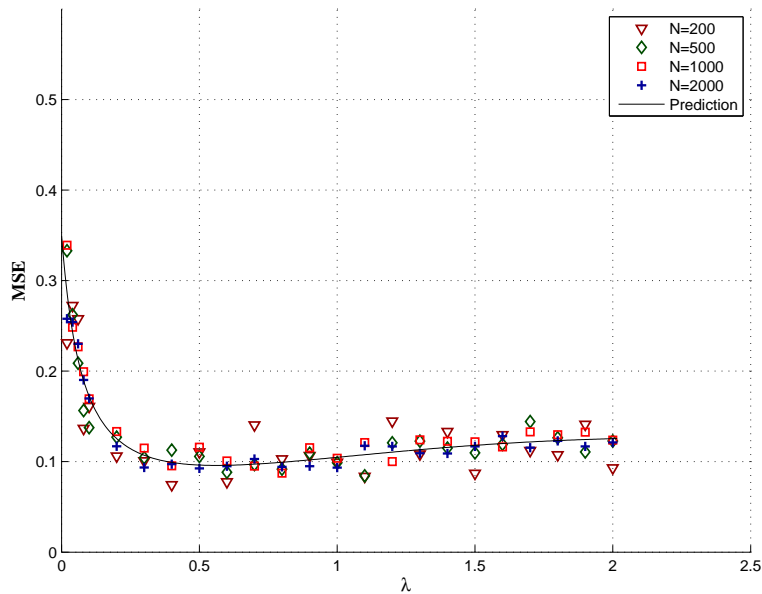}
 \caption{As in Fig.~\ref{fig:GaussianMatrices}, but  the
measurement matrix $A$ has i.i.d. entries that are equal to $\pm1/\sqrt{m}$
with equal probabilities.}\label{fig:PM1Matrices}
\end{figure}
%
%
\subsection{A decoupling principle}

There exists a suggestive interpretation of the state evolution result
in Theorem \ref{prop:state-evolution}, as well as of the 
scaling limit of the LASSO established in Theorem \ref{thm:Risk}:
\emph{
The estimation problem in the vector model $y=Ax+w$
reduces --asymptotically-- to $n$ uncoupled
scalar estimation problems $\widetilde{y}_i = x_i +\widetilde{w}_i$.}
However the noise variance is increased from $\sigma^2$ to
$\tau^2_t$ (or $\tauinf^2$ in the case of the LASSO),
due to `interference' between the original coordinates:
\begin{eqnarray}
y=Ax+w \;\;\;\;\;\; \Leftrightarrow \;\;\;\;
\left\{\begin{array}{l}
\widetilde{y}_1 = x_1 +\widetilde{w}_1\\
\widetilde{y}_2 = x_2 +\widetilde{w}_2\\
\vdots\\
\widetilde{y}_n = x_n +\widetilde{w}_n
\end{array}\right.\, .
\end{eqnarray}
An analogous phenomenon is well known in statistical physics 
and probability theory and takes sometimes the name of 
`correlation decay' \cite{WeitzCrit,Gamarnik07,MezardMontanari}. 
In the context of CDMA system analysis via replica method,
the same phenomenon was also called `decoupling principle'
\cite{TanakaCDMA,GuoVerdu}.

Notice that the AMP algorithm gives a precise realization of this
decoupling principle, since for each $i\in [n]$,
and for each number of iterations $t$, it produces an estimate,
namely $(x^t+A^Tr^t)_i$ that can be considered a realization
of the observation $\widetilde{y}_i$ above.
Indeed Theorem \ref{prop:state-evolution} (see also discussion below
the theorem) states that $(x^t+A^Tr^t)_i=x_i+\widetilde{w}_i$
with $\widetilde{w}_i$ asymptotically Gaussian with 
mean $0$ and variance $\tau_t^2$.

The fact that observations of distinct coordinates are asymptotically 
decoupled is stated precisely below.
\begin{corollary}[Decoupling principle, \cite{BM-MPCS-2010}]
\label{coro:Decoupling}
Under the assumption of Theorem \ref{prop:state-evolution}, fix  $\ell\ge 2$,
let $\psi:\reals^{2\ell}\to\reals$ be any Lipschitz function,
and denote by ${\sf E}$ expectation with respect to a
uniformly random subset of distinct indices
$J(1),\dots,J(\ell)\in [n]$. 

Further, for some fixed $t>0$, let $\widetilde{y}^t = x^t+A^Tr^t\in\reals^n$. 
Then, almost surely
\begin{eqnarray*}
\lim_{n\to\infty}{\sf E}\psi(\widetilde{y}^t_{J(1)},\dots,
\widetilde{y}^t_{J(\ell)},
x_{J(1)},\dots,x_{J(\ell)})
=\E\big\{\psi\big(
X_{0,1}+\tau_t Z_1,\dots,X_{0,\ell}+\tau_t Z_{\ell},X_{0,1},\dots,
X_{0,\ell}\big)\big\} ,\phantom{a}\label{eq:Kconv}
\end{eqnarray*}
for  $X_{0,i}\sim p_{0}$ and $Z_i\sim \normal(0,1)$,
$i=1,\dots,\ell$ mutually independent.
\end{corollary}

%
%
\subsection{An heuristic derivation of state evolution} 

The state evolution recursion has a simple heuristic
description that is useful to present here since it
clarifies the difficulties  involved in the proof.
In particular, this description brings up the key role
played by the `Onsager term' appearing in Eq.~(\ref{eq:dmm2}) \cite{DMM09}.

Consider again the recursion (\ref{eq:dmm1}),
(\ref{eq:dmm2}) but
introduce the following three modifications:
$(i)$ Replace the random matrix $A$ with a new
independent copy $A(t)$ at each iteration $t$;
$(ii)$ Correspondingly replace the observation vector $y$
with $y^t=A(t)x+w$;
$(iii)$ Eliminate the last term in the update equation for $r^t$.
We thus get the following dynamics:
\begin{align}
x^{t+1}& = \eta(A(t)^Tr^t+x^t;\theta_t)\, ,\\
r^{t}& = y^t-A(t)x^t\, ,
\end{align}
where $A(0),A(1),A(2),\dots$ are i.i.d. matrices of dimensions $m\times n$
with i.i.d. entries $A_{ij}(t)\sim \normal(0,1/m)$.
(Notice that, unlike in the rest of the article, we use here the
argument of $A$ to denote the iteration number, and not the matrix
dimensions.)

This recursion is most
conveniently written by eliminating $r^t$:
\begin{align}
x^{t+1}& = \eta\big(A(t)^Ty^t+(\identity-A(t)^TA(t))x^t;\theta_t\big)\, ,\nonumber\\
       & = \eta\big(x+A(t)^Tw+B(t)(x^t-x);\theta_t\big)\, ,
\label{eq:ModifiedRecursionNoAlg}
\end{align}
where we defined $B(t) = \identity-A(t)^{T}A(t)\in\reals^{n\times n}$.
Let us stress that this recursion does not correspond to any concrete algorithm,
since the matrix $A$ changes from iteration to iteration.
It is nevertheless useful for developing intuition.

Using the central limit theorem, it is easy
to show that each entry of $B(t)$ is approximately
normal, with  zero mean and variance $1/m$.
Further, distinct entries are approximately pairwise independent.
Therefore, if we let
$\ttau_t^2 = \lim_{n\to\infty}\|x^t-x\|_2^2/n$, we obtain that $B(t)(x^t-x)$
converges to a vector with i.i.d. normal entries with $0$ mean and
variance $n\ttau_t^2/m = \ttau_t^2/\delta$. Notice that this is true
because $A(t)$ is independent of $\{A(s)\}_{1\le s\le t-1}$ and,
in particular, of $(x^t-x)$.

Conditional on $w$, $A(t)^Tw$ is a vector of i.i.d. normal
entries with mean $0$  and variance $(1/m)\|w\|_2^2$
which converges by assumption to $\sigma^2$.
A slightly longer exercise shows that these entries are
approximately independent from the ones of $B(t)(x^t-x_0)$.
Summarizing, each entry of the vector in the argument
of $\eta$ in Eq.~(\ref{eq:ModifiedRecursionNoAlg})
converges to $X_0+\tau_t Z$ with $Z\sim\normal(0,1)$
independent of $X_0$, and
\begin{align}
\tau_t^2 &= \sigma^2+\frac{1}{\delta}\ttau_t^2\, ,\label{eq:FirstHeuristic}\\
\ttau_t^2 & =\lim_{n\to\infty}\frac{1}{n}\|x^t-x\|_2^2\, .\nonumber
\end{align}
On the other hand, by Eq.~(\ref{eq:ModifiedRecursionNoAlg}),
each entry of $x^{t+1}-x$ converges to $\eta(X_0+\tau_t\, Z;\theta_t)-X_0$,
and therefore
\begin{align}
\ttau_{t+1}^2 & =\lim_{n\to\infty}\frac{1}{n}\|x^{t+1}-x\|_2^2
=\E\big\{[\eta(X_0+\tau_t\, Z;\theta_t)-X_0]^2\big\}\, .\label{eq:SecondHeuristic}
\end{align}
Using together Eq.~(\ref{eq:FirstHeuristic}) and (\ref{eq:SecondHeuristic})
we finally obtain the state evolution recursion, Eq.~(\ref{eq:1-dim-SE}).

We conclude that state evolution would hold if the matrix $A$
was drawn independently from the same Gaussian distribution at
each iteration. In the case of interest, $A$ does not change across
iterations, and the above argument falls apart because
$x^t$ and $A$ are dependent. This dependency is non-negligible even
in the large system limit $n\to\infty$. This point can be clarified by
considering the IST algorithm given by Eqs.~(\ref{eq:ist1}),
(\ref{eq:ist2}).
Numerical studies of iterative soft thresholding
\cite{DonohoMalekiTuned,DMM09} show that its behavior is dramatically
different from the one of AMP and in particular
\emph{state evolution does not hold for IST}, even in the large
system limit.

This is not a surprise: the correlations between $A$ and $x^t$
simply cannot be neglected.
On the other hand, adding the Onsager term leads to an asymptotic cancelation
of these correlations. As a consequence, state evolution
holds for the AMP iteration.
%
%
\subsection{The noise sensitivity phase transition}
\label{sec:NSPT}
\begin{figure}
\begin{center}
\includegraphics[width=11cm,angle=0]{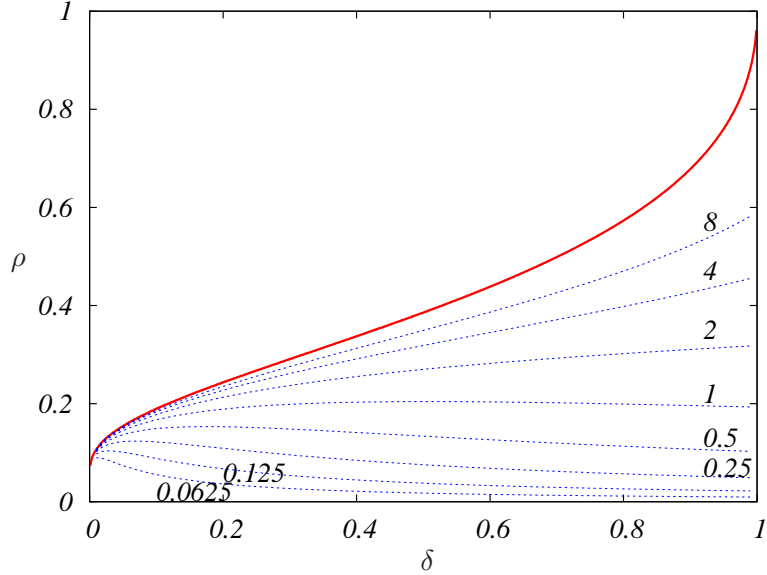}
\put(-145,-5){$\delta$}
\put(-300,110){$\rho$}
\end{center}
\caption{{\small Noise sensitivity phase transition in the plane 
$(\delta,\rho)$ (here $\delta=m/n$ is the undersampling ratio and
$\rho=\|x\|_0/m$ is the number of non-zero coefficients per measurement).
Red line: The phase transition boundary $\rho = \rho_{\rm c}(\delta)$.
Blue lines: Level curves for the LASSO minimax $M^*(\delta,\rho)$.
Notice that $M^*(\delta,\rho)\uparrow\infty$ as $\rho\uparrow
\rho_{\rm c}(\delta)$.}}
\label{fig:NSPT}
\end{figure}
The formalism developed so far allows to extend the minimax
analysis carried out in the scalar case in Section \ref{sec:Scalar}
to the vector estimation problem
\cite{NSPT}. We define the LASSO mean square error per
coordinate when the empirical distribution of the signal 
converges to $p_0$,  as
\begin{eqnarray}
\laMSE(\sigma^2;p_0,\lambda) = \lim_{n\to\infty}
\frac{1}{n}\,\E\big\{\|\hx(\lambda)-x\|_2^2
\big\}\, ,
\end{eqnarray}
where the limit is taken along a converging sequence.
This quantity can be computed using Theorem \ref{thm:Risk}
for any specific distribution $p_0$.

We consider again the sparsity class $\cF_{\ve}$ with $\ve=\rho\delta$.
Hence $\rho = \|x\|_0/m$ measures the number of non-zero 
coordinates per measurement. 
Taking the worst case MSE over this class,
and then the minimum over the regularization parameter 
$\lambda$, we get a result that depends on $\rho$, $\delta$,
as well as on the noise level $\sigma^2$. The dependence on $\sigma^2$
must be linear because the class $\cF_{\rho\delta}$ is scale invariant,
and we obtain therefore
\begin{eqnarray}
\inf_{\lambda} \sup_{p_0\in \cF_{\rho\delta}}\, \laMSE(\sigma^2;p_0,\lambda)
= M^*(\delta,\rho) \, \sigma^2\, ,
\end{eqnarray}
for some function $(\delta,\rho)\mapsto M^*(\delta,\rho)$.
We call this the LASSO minimax risk. It can be interpreted as the sensitivity
(in terms of mean square error) of the LASSO estimator to noise in the 
measurements.

It is clear that the prediction for $\laMSE(\sigma^2;p_0,\lambda)$
provided by Theorem \ref{thm:Risk} can be used to characterize 
the LASSO minimax risk. What is remarkable is that the resulting
formula is so simple. 
\begin{theorem}[\cite{NSPT}]\label{thm:NSPT}
Assume the hypotheses of Theorem \ref{thm:Risk}, and 
recall that $M^{\#}(\ve)$ denotes the soft thresholding minimax
risk over the class $\cF_{\ve}$ cf. Eqs.~(\ref{eq:M1def}),
(\ref{eq:Mhashdef}). Further let $\rho_{\rm c}(\delta)$
be the unique solution of $\rho = M^{\#}(\rho\delta)$.

Then for any $\rho<\rho_{\rm c}(\delta)$ the 
$\LASSO$ minimax risk is bounded
and given by
\begin{eqnarray}
M^*(\delta,\rho) = \frac{M^{\#}(\rho\delta)}{1-M^{\#}(\rho\delta)/\delta}\, .
\end{eqnarray}
Viceversa, for any $\rho\ge\rho_{\rm c}(\delta)$, we have
 $M^*(\delta,\rho) =\infty$.
\end{theorem}

Figure \ref{fig:NSPT} shows the location of the noise sensitivity 
boundary $\rho_{\rm c}(\delta)$ as well as the level lines of
$M^*(\delta,\rho)$ for $\rho<\rho_{\rm c}(\delta)$. Above 
$\rho_{\rm c}(\delta)$ the LASSO MSE is not uniformly bounded 
in terms of the measurement noise $\sigma^2$. Other estimators
(for instance one step of soft thresholding) can offer
better stability guarantees in this region. 

One remarkable fact is that the phase boundary $\rho=\rho_{\rm c}(\delta)$
coincides with the phase transition for $\ell_0-\ell_1$
equivalence derived earlier by Donoho \cite{DonohoCentrally} 
on the basis of random polytope geometry
results by  Affentranger-Schneider \cite{Affentranger}. The same 
phase transition was further studied in a series of papers by 
Donoho, Tanner and coworkers \cite{DoTa05,DoTa08}, 
in connection with the noiseless estimation
problem. For $\rho<\rho_{\rm c}$ estimating $x$ by $\ell_1$-norm minimization
returns the correct signal with high probability (over the 
choice of the random matrix $A$). For $\rho>\rho_{\rm c}(\delta)$,
$\ell_1$-minimization fails. 

Here this phase transition is derived from a completely 
different perspective as a special case of a stronger result.
We indeed use a new method --the state
evolution analysis of the AMP algorithm-- which offers
quantitative information about the noisy case as well,
namely it allows 
to compute the value of $M^*(\delta,\rho)$ for $\rho<\rho_{\rm c}(\delta)$.
Within the present approach, the line $\rho_{\rm c}(\delta)$
admits a very simple expresson. In parametric form, it is given
by
\begin{eqnarray}
\delta & = &\frac{2\phi(\alpha)}{\alpha+2(\phi(\alpha)-\alpha\Phi(-\alpha))}
\, ,\label{eq:Parametric1}\\
\rho & = &1-\frac{\alpha\Phi(-\alpha)}{\phi(\alpha)}\, ,\label{eq:Parametric2}
\end{eqnarray}
where $\phi$ and $\Phi$ are the Gaussian density and Gaussian distribution 
function, and $\alpha\in [0,\infty)$ is the parameter. Indeed 
$\alpha$ has a simple and practically important interpretation as well. 
Recall that the AMP algorithm uses a sequence of thresholds
$\theta_t = \alpha\htau_t$, cf. Eqs.~(\ref{eq:ThresholdChoice1})
and (\ref{eq:ThresholdChoice2}). How should the parameter $\alpha$ 
be fixed? A very simple prescription is obtained in the noiseless case.
In order to achieve exact reconstruction for all $\rho<\rho_{\rm c}(\delta)$
for a given an undersampling ratio $\delta$,  $\alpha$
should be such that $(\delta,\rho_{\rm c}(\delta)) = (\delta(\alpha),
\rho(\alpha))$ with functions $\alpha\mapsto\delta(\alpha)$,
$\alpha\mapsto\rho(\alpha)$ defined as in Eq.~(\ref{eq:Parametric1}),
(\ref{eq:Parametric2}). In other words, this parametric expression 
yields each point of the phase boundary as a function of the threshold
parameter used to achieve it via AMP.
%
%
\subsection{On universality}
\label{sec:Universality}

\begin{figure}
\centering
  \includegraphics[width=4.in]{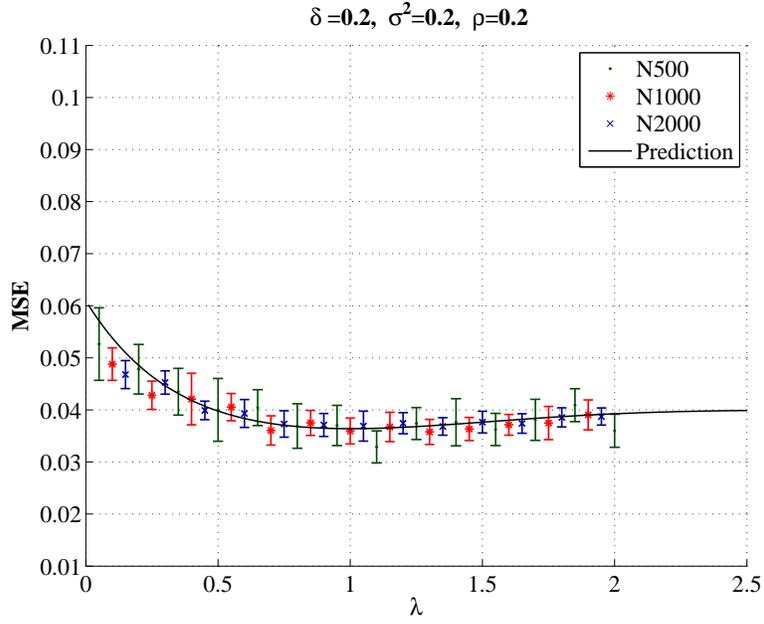}
 \caption{Mean square error for as a function of the regularization parameter 
$\lambda$ for a partial Fourier matrix (see text).
The noise variance is $\sigma^2=0.2$, the undersampling factor $\delta=0.2$
and the sparsity ratio $\rho=0.2$.  
Data points are obtained by averaging over $20$ realizations, 
and error bars are $95\%$ confidence intervals. The continuous line 
is the prediction of Theorem \ref{thm:Risk}.}
\label{fig:FourierMatrix}
\end{figure}

\begin{figure}
\centering
\includegraphics[width=4.in]{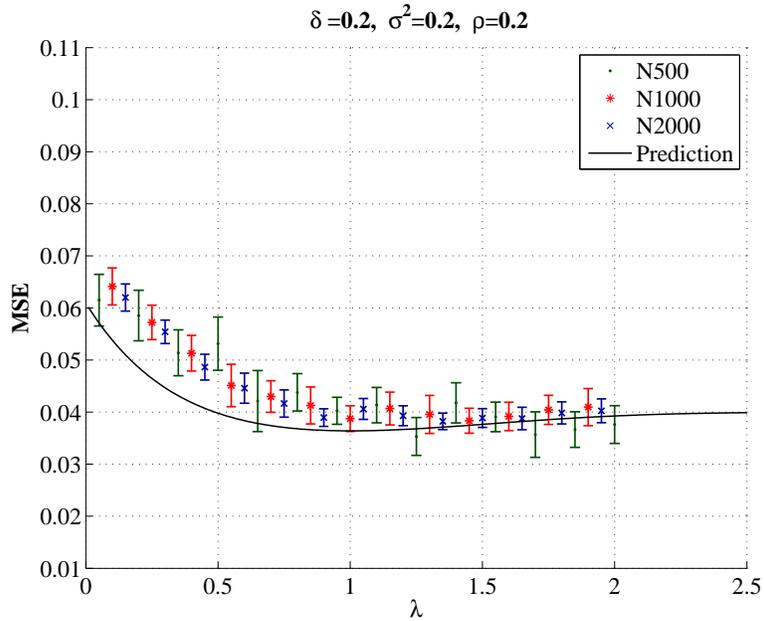}
 \caption{As in Fig.~\ref{fig:FourierMatrix}, but
for a measurement matrix $A$ which 
models the analog-to-digital converter of \cite{TroppADC}.}
\label{fig:ADCMatrix}
\end{figure}

The main results presented
in this section, namely Theorems \ref{prop:state-evolution},
\ref{thm:Risk} and \ref{thm:NSPT}, are proved for 
measurement matrices with i.i.d. Gaussian entries. 
As stressed above, it is expected that the same results hold for 
a much broader class of matrices. In particular, they should extend
to matrices with i.i.d. or weakly correlated entries.
For the sake of clarity, it is useful to put forward a formal 
conjecture, that generalizes Theorem 
\ref{thm:Risk}.
\begin{conj}\label{conj:Risk}
Let $\{x(n), w(n), A(n)\}_{n\in\naturals}$ be a converging sequence of
instances with the entries of $A(n)$ i.i.d. with mean 
 $\E\{A_{ij}\}=0$, variance
$\E\{A_{ij}^2\}=1/m$ and such that 
$\E\{A_{ij}^6\}\le C/m$ for some fixed constant $C$. 
Denote by $\hx(\lambda)$ the \LASSO\, estimator for
instance $(x(n), w(n), A(n))$, with $\sigma^2,\lambda> 0$, and
let $\psi:\reals\times\reals\to \reals$ be a pseudo-Lipschitz function.
Then, almost surely
\begin{eqnarray}\label{eq:UniversalFormula}
\lim_{n\to\infty}\frac{1}{n}\sum_{i=1}^n\psi
\big(\hx_{i},x_{i}\big) = \E\Big\{\psi\big(\eta(X_0+\tau_* Z;\theta_*),X_0\big)\Big\}\, ,
\end{eqnarray}
where $Z\sim\normal(0,1)$ is independent of $X_0\sim p_{0}$,
$\tau_*=\tau_*(\alpha(\lambda))$ and $\theta_*=\alpha(\lambda)
\tau_*(\alpha(\lambda))$ are given by the same formulae
holding for Gaussian matrices, cf. Section \ref{sec:Results}.
\end{conj}
The conditions formulated in this conjecture are motivated
by the universality result in \cite{KM-2010}, that provides
partial evidence towards this claim. 
Simulations (see for instance Fig.~\ref{fig:PM1Matrices} and
\cite{OurLASSO_Exp}) strongly support this claim.

While proving Conjecture \ref{conj:Risk} is an outstanding mathematical 
challenge, many measurement models of interest do not fit the 
i.i.d. model. Does the theory developed in these section say anything 
about such measurements? Systematic numerical simulations 
\cite{NSPT,OurLASSO_Exp} reveal that, even for highly structured matrices,
the same formula \ref{eq:UniversalFormula} is either 
surprisingly close to the actual empirical performances.

As an example, Fig.~\ref{fig:FourierMatrix} presents the empirical
mean square error for a partial Fourier measurement matrix $A$,
as a function of the regularization parameter $\lambda$.
The matrix is obtained by subsampling
the rows of the $N\times N$ Fourier matrix $F$, with entries
$F_{ij}=e^{2\pi ij\sqrt{-1}}$. More precisely we sample $n/2$ rows of
$F$ with replacement, construct two rows of $A$ by taking real and 
imaginary part, and normalize the columns of the resulting matrix.

Figure \ref{fig:ADCMatrix} presents analogous results for 
the random demodulator matrix  which is at the core
of the analog-to-digital converter (ADC) of \cite{TroppADC}.
Schematically, this is obtained by normalizing the columns
of $\widetilde{A} = HDF$, with 
$F$ a Fourier matrix, $D$ a random diagonal matrix with 
$D_{ii}\in\{+ 1,-1\}$ uniformly at random, and
$H$ an `accumulator':
\begin{eqnarray*}
H = \left[\begin{array}{cccc}
1111&  & & \\
&1111 & & \\
&& \cdots &\\
&&& 1111\\
\end{array}
\right]\, .
\end{eqnarray*}

Both these examples show good agreement between the asymptotic
prediction provided by Theorem \ref{thm:Risk}
and the empirical mean square error. Such an agreement 
is surprising given that in both cases the measurement matrix
is generated with a small amount of randomness, compared to
a Gaussian matrix. For instance, the ADC matrix only requires $n$
random bits. Although statistically significant discrepancies can be
observed (cf. for instance Fig.~\ref{fig:ADCMatrix}),
the present approach provides \emph{quantitative} predictions
of great interest for design purposes. For a more systematic
investigation, we refer to \cite{OurLASSO_Exp}.
%
%
\subsection{Comparison with other analysis approaches}

The analysis presented here is significantly different from
more standard approaches. We derived an \emph{exact}
characterization for the high-dimensional limit 
of the LASSO estimation problem under the assumption of 
converging sequences of random sensing matrices.

Alternative approaches assume an
appropriate `isometry', or `incoherence' condition to hold for $A$.
Under this condition upper bounds are proved for the 
mean square error. For instance  Candes, Romberg and Tao
\cite{CandesStable} prove that the mean square error is bounded
by $C\sigma^2$ for some constant $C$. 
 Work by Candes and Tao \cite{Dantzig} on the analogous
\emph{Dantzig selector}, upper bounds
the mean square error by $C\sigma^2 (k/n)\log n$,
with $k$ the number of non-zero entries of the signal $x$.

These type of results are very robust but present two limitations:
$(i)$ They do not allow to distinguish reconstruction methods
that differ by a constant factor (e.g. two different values of $\lambda$);
$(ii)$ The restricted isometry condition (or analogous ones)
is quite restrictive. For instance, it holds for random matrices 
only under very strong sparsity assumptions.
These restrictions are intrinsic to the worst-case
point of view developed in \cite{CandesStable,Dantzig}.

Guarantees have been proved for correct support recovery in
\cite{Zhao}, under an  incoherence assumption on $A$.
While support recovery is an interesting conceptualization for
some applications (e.g. model selection), the metric considered
in the present paper (mean square error) provides complementary
information and is quite standard
in many different fields.

Close to the spirit of the treatment presented here,
\cite{Goyal}
derived expressions for the mean square error under
the same model considered here. Similar results
were presented recently in \cite{KabashimaTanaka,BaronGuoShamai}.
These papers argue that a sharp asymptotic characterization
of the LASSO risk can provide valuable guidance in practical
applications. 
Unfortunately, these results were non-rigorous and were obtained
through the famously powerful `replica method' from statistical physics
\cite{MezardMontanari}.
The approach discussed here offers two advantages
over these recent
developments: $(i)$ It is completely \emph{rigorous}, thus putting
on a firmer basis this line of research;
$(ii)$ It is \emph{algorithmic} in that
the LASSO mean square error is shown to be equivalent to the one
achieved by a low-complexity message passing algorithm.

Finally, recently random models for the measurement matrix have been studied
in \cite{CandesPlan1,CandesPlan2}. The approach developed in these papers allows
to treat matrices that do not necessarily satisfy the restricted
isometry property or similar conditions, and applies to a general 
class of random matrices $A$ with i.i.d. rows. On the other
hand, the resulting bounds are not asymptotically sharp.
%
%
\section{Generalizations}
\label{sec:Generalizations}

The single most important advantage of the point of view based on graphical
models is that it offers a unified disciplined approach to
exploit structural information on the signal $x$. 
The use of such information can dramatically reduce the number of 
required compressed sensing measurements. 

`Model-based' compressed sensing \cite{ModelCS}
provides a general framework for specifying such information.
However, it focuses on `hard' combinatorial information
about the signal. Graphical models are instead a rich 
language for specifying `soft' dependencies or constraints,
and more complex models. These might include combinatorial constraints,
but vastly generalize them. 
Also, graphical models come with an algorithmic arsenal that
can be applied to leverage the potential of such more complex 
signal models.

Exploring such potential generalizations is --to a large extent-- a future 
research program which is still in its infancy. Here we will only discuss 
a few examples. 
%
%
\subsection{Structured priors\dots}
Block-sparsity is a simple example of combinatorial signal structure.
We decompose the signal as
$x = (x_{B(1)},$ $x_{B(2)},$ $\dots,$ $x_{B(\ell)})$ 
where $x_{B(i)}\in \reals^{n/\ell}$ is a block
for $\ell\in\{1,\dots,\ell\}$.
Only a  fraction $\ve\in (0,1)$ of the blocks is non-vanishing.
This type of model naturally arises in many applications:
for instance the case $\ell=n/2$ (blocks of size $2$)
can model signals with complex-valued entries. Larger 
blocks can correspond to shared sparsity patterns among many vectors,
or to clustered sparsity.

\begin{figure}
\begin{center}
\includegraphics[width=4cm,angle=90]{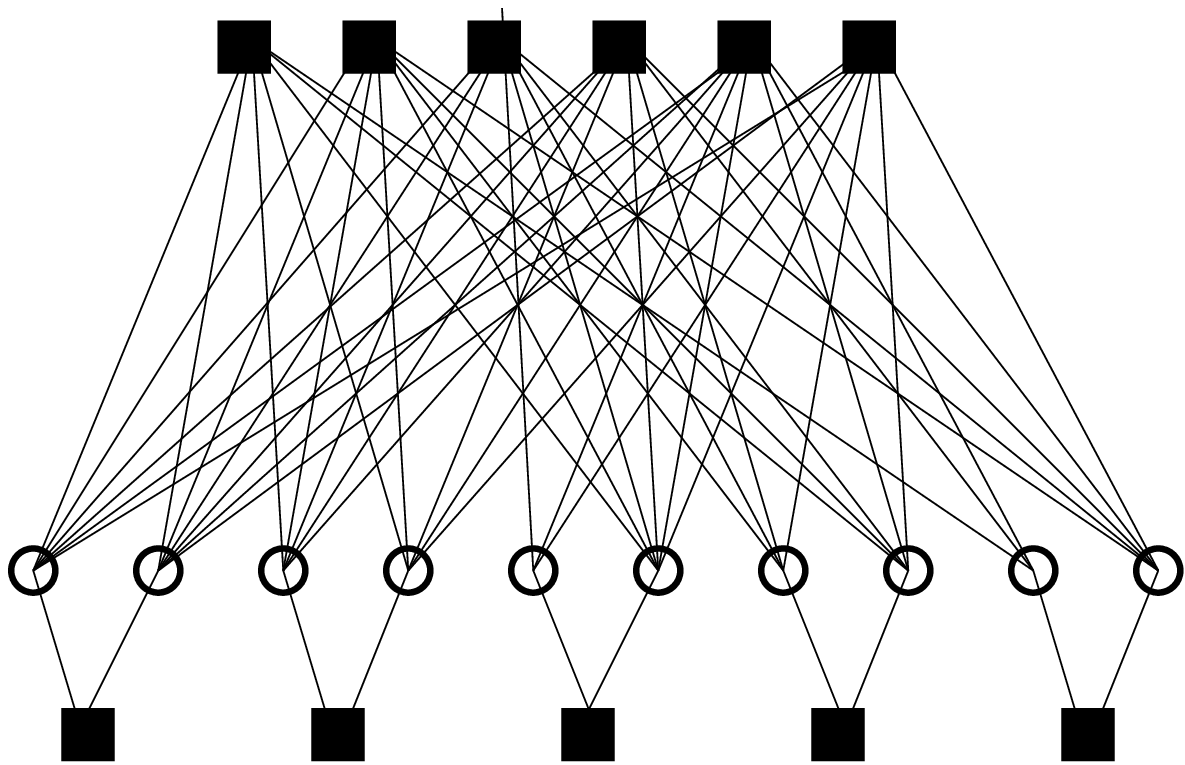}
\phantom{AAAAA}
\includegraphics[width=4cm,angle=90]{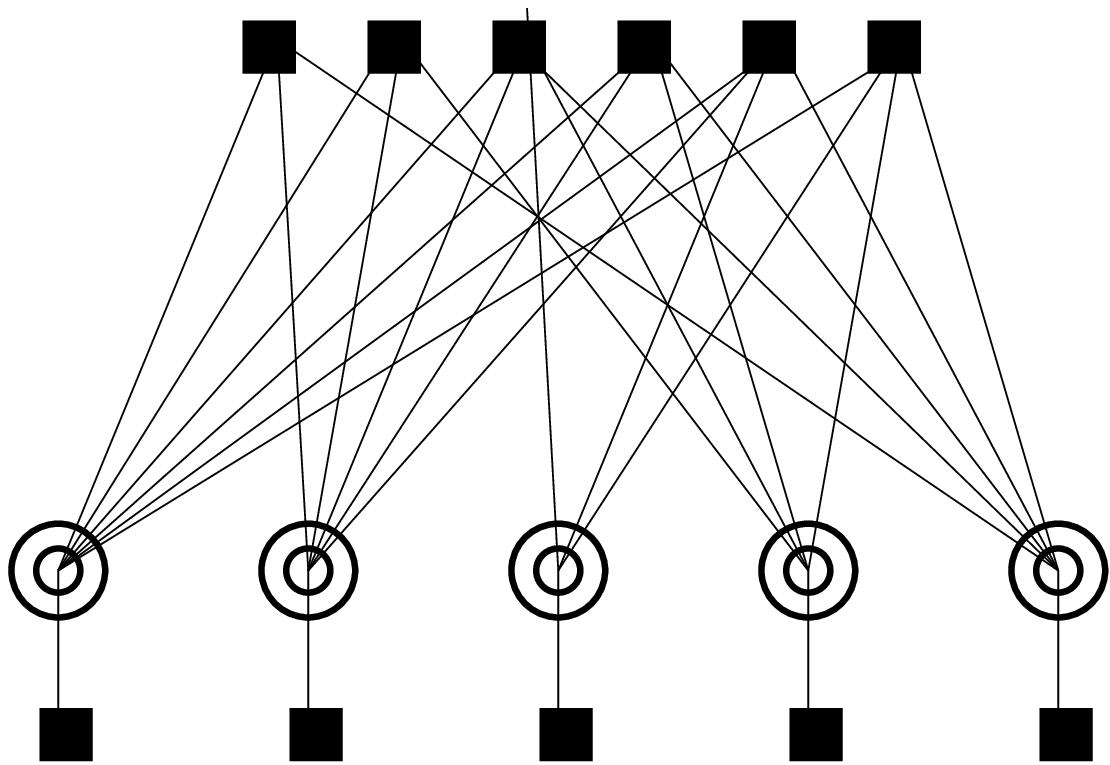}
\put(-403,38){$x_1$}
\put(-215,38){$x_n$}
\put(-184,40){$x_{B(1)}$}
\put(-8,40){$x_{B(\ell)}$}
\put(-135,117){$1$}
\put(-38,117){$m$}
\put(-358,117){$1$}
\put(-262,117){$m$}
\end{center}
\caption{{\small Two possible graphical representation 
of the block-sparse compressed sensing model
(and corresponding cost function 
(\ref{eq:GeneralModelFinal})). 
Upper squares correspond to measurements $y_a$, $a\in [m]$,
and lower squares to the block sparsity constraint
(in this case blocks have size $2$). 
On the left, circles correspond to variables $x_i\in\reals$,
$i\in [n]$. On they right, double circles correspond to blocks $x_{B(i)}
\in \reals^{n/\ell}$, $i\in [\ell]$.
}}
\label{fig:FactorGraphBlock}
\end{figure}
It is customary in this setting to replace the LASSO 
cost function with the following
\begin{align}
\cost^{{\rm Block}}_{A,y}(z) & \equiv \frac{1}{2}\|y-Az\|_2^2 +  \lambda
\sum_{i=1}^{\ell}\|z_{B(i)}\|_2\, .\label{eq:CostBlock}
\end{align}
The block-$\ell_2$ regularization promotes block sparsity. 
Of course, the new regularization can be interpreted in terms of a 
new assumed prior that factorizes over blocks. 

Figure \ref{fig:FactorGraphBlock} reproduces two possible graphical 
structures that encode the block-sparsity constraint. 
In the first case, this is modeled explicitly as a 
constraint over blocks of variable nodes, 
each block comprising $n/\ell$ variables. In the 
second case, blocks correspond explicitly to variables 
taking values in $\reals^{n/\ell}$.
Each of these graphs dictates a somewhat different 
message passing algorithm.

An approximate 
message passing algorithm suitable for this case is developed in
\cite{DonohoBlock}. Its analysis allows to generalize
$\ell_0-\ell_1$ phase transition curves reviewed
in Section \ref{sec:NSPT}
to the block sparse 
case. This quantifies precisely the benefit of minimizing
(\ref{eq:CostBlock}) over simple $\ell_1$ penalization.

\vspace{0.2cm}

As mentioned above, for a large class of signals sparsity is not uniform:
some subsets of entries are sparser than others.
Tanaka and Raymond \cite{TanakaRaymond-2010}, and 
Som, Potter and Schniter and \cite{Schniter-NonUniform-2010} studied
the case of signals with multiple level of sparsity. The simplest
example consists of a signal $x = (x_{B(1)},x_{B(2)})$, where 
$x_{B(1)}\in\reals^{n_1}$,
$x_{B(2)}\in\reals^{n_2}$, $n_1+n_2=n$. Block $i\in\{1,2\}$ has 
a fraction $\ve_i$ of non-zero entries, with $\ve_1\neq \ve_2$.
In the most complex case, one can consider a general factorized 
prior
\begin{eqnarray*}
p(\de x) = \prod_{i=1}^{n}p_{i}(\de x_i)\, ,
\end{eqnarray*}
where each $i\in[n]$ has a different sparsity parameter
$\ve_i\in (0,1)$, and $p_i\in\cF_{\ve_i}$.
In this case it is natural to use a weighted--$\ell_1$ 
regularization, i.e. to minimize
\begin{align}
\cost^{{\rm weight}}_{A,y}(z) & \equiv \frac{1}{2}\|y-Az\|_2^2 +  \lambda
\sum_{i=1}^{n}w_i\, |z_{i}|\, ,
\end{align}
for a suitable choice of the weights $w_1,$ $\dots,$ $w_n\ge 0$.
The paper \cite{TanakaRaymond-2010} studies the case $\lambda\to 0$
(equivalent to minimizing $\sum_iw_i|z_i|$ subject to $y = Az$),
using non-rigorous statistical mechanics techniques that are  equivalent to
the state evolution approach presented here.
Within a high-dimensional limit, it determines 
optimal tuning of the parameters $w_i$, for given sparsities $\ve_i$.
The paper \cite{Schniter-NonUniform-2010}
follows instead the state evolution approach explained in the present chapter.
The authors develop a suitable AMP iteration and compute
the optimal thresholds to be used by the algorithm.
These are in correspondence with the optimal weights $w_i$
mentioned above, and can be also interpreted within the
minimax framework developed in the previous pages. 

\vspace{0.2cm}

\begin{figure}
\begin{center}
\includegraphics[width=6cm,angle=90]{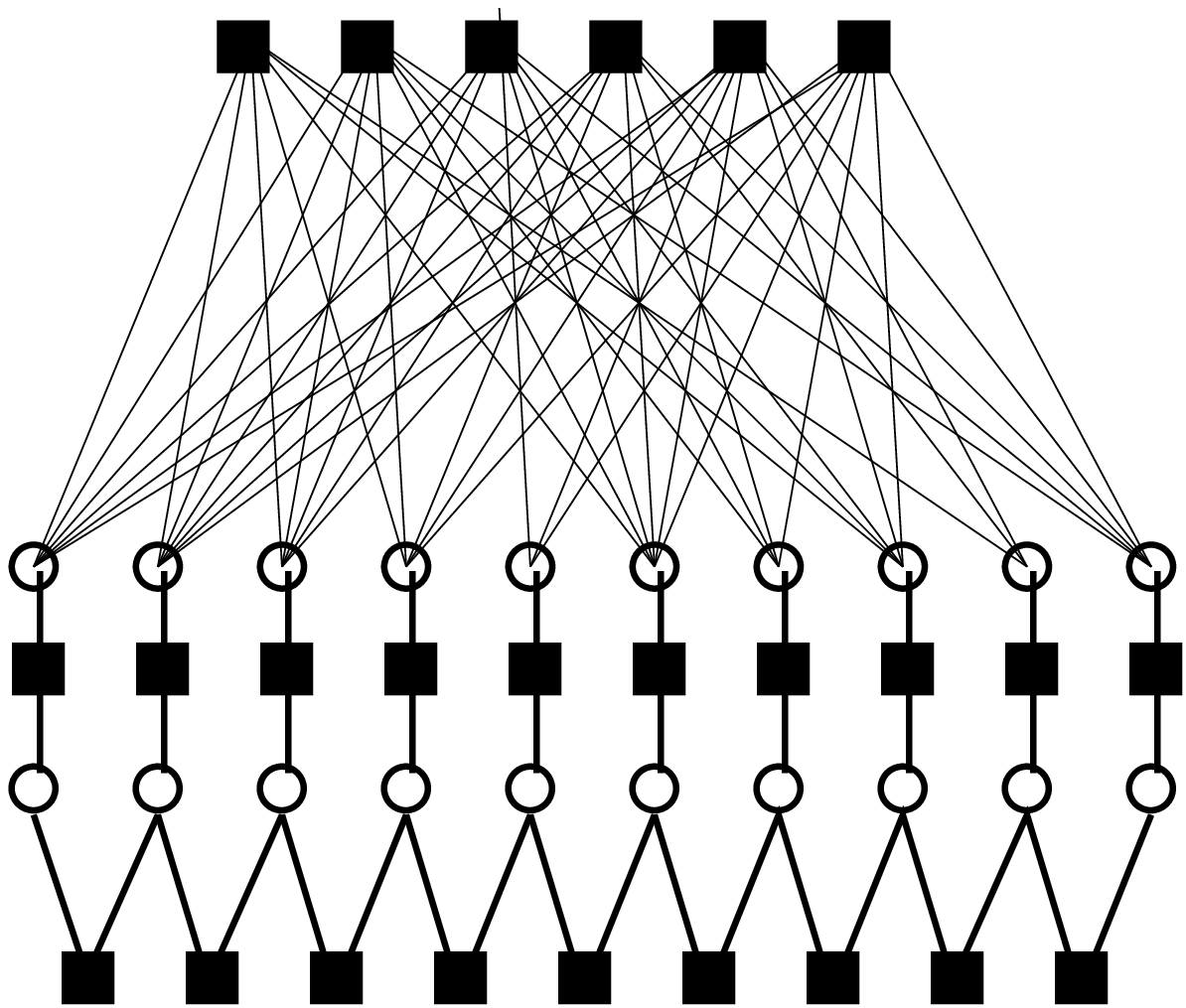}
\put(-214,75){$x_1$}
\put(0,75){$x_n$}
\put(-214,40){$s_1$}
\put(0,40){$s_n$}
\put(-165,172){$1$}
\put(-55,172){$m$}
\end{center}
\caption{{\small Graphical model for compressed sensing of
signals with clustered support. The support structure is described by 
an Hidden Markov Model comprising the lower
factor nodes (filled squares) and variable nodes (empty circles).
Upper variable nodes correspond to the signal entries $x_i$,
$i\in [n]$, and upper factor nodes to the measurements 
$y_a$, $a\in [m]$.}}
\label{fig:FactorGraphMC}
\end{figure}

The graphical model framework is particularly convenient
for exploiting prior information that is probabilistic in 
nature, see in particular \cite{CevherMRF,CevherReview}. 
A prototypical example was studied by Schniter \cite{SchniterTurbo} 
who considered the case in
which the signal $x$ is generated by an Hidden Markov Model
(HMM). As for the block-sparse model, this can be used to model 
signals in which the non-zero coefficients are clustered, although in this
case one can accomodate greater stochastic variability of the
cluster sizes.

In the simple case studied in detail in 
\cite{SchniterTurbo}, the underlying Markov chain 
has two states indexed by $s_i\in \{0,1\}$, and
\begin{eqnarray}
p(\de x)= \sum_{s_1,\dots, s_n} 
\Big\{\prod_{i=1}^np(\de x_i|s_i)
\cdot \prod_{i=1}^{n-1}p(s_{i+1}|s_i)\, \cdot p_1(s_1)\Big\}\, ,
\end{eqnarray} 
where $p(\,\cdot\,|0)$ and $p(\,\cdot\,|1)$ belong to two different
sparsity classes $\cF_{\ve_0}$, $\cF_{\ve_1}$.
For instance one can consider the case in which $\ve_0=0$ and
$\ve_1=1$, i.e. the support of $x$ coincides with the subset of coordinates
such that $s_i=1$.

\begin{figure}
\begin{center}
\includegraphics[width=6.5cm,angle=90]{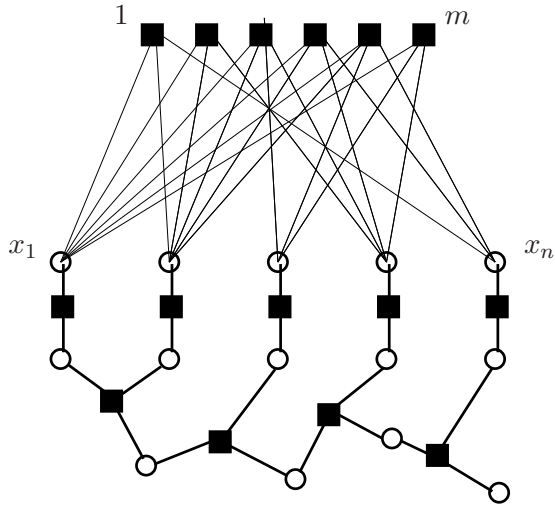}
\put(-190,95){$x_1$}
\put(5,95){$x_n$}
\put(-150,182){$1$}
\put(-25,182){$m$}
\end{center}
\caption{{\small Graphical model for compressed sensing of
signals with tree-structured prior. The support structure is
a tree graphical model, comprising factor nodes and variable nodes
in the lower part of the graph.
Upper variable nodes correspond to the signal entries $x_i$,
$i\in [n]$, and upper factor nodes to the measurements 
$y_a$, $a\in [m]$.}}
\label{fig:TreePrior}
\end{figure}

Figure \ref{fig:FactorGraphMC} reproduces the graphical structure 
associated with this type of models. This
can be partitioned in two components: a bipartite
graph corresponding to the compressed sensing measurements 
(upper part in Fig.~ \ref{fig:FactorGraphMC}) and
a chain graph corresponding to the Hidden Markov Model
structure of the prior (lower part in Fig.~ \ref{fig:FactorGraphMC}).

Reconstruction was performed in \cite{SchniterTurbo} using
a suitable generalization of AMP. Roughly speaking, 
inference is performed in the upper half of the graph using 
AMP and in the lower part using the standard forward-backward 
algorithm. 
Information is exchanged across the two component in a 
way that is very similar to what happens in turbo codes \cite{RiU08}.

\vspace{0.2cm}

The example of HMM priors clarifies the usefulness of
the graphical model structure in eliciting tractable
substructures in the probabilistic model and hence leading to
natural iterative algorithms. For an HMM prior, inference
can be performed efficiently because the underlying graph is a 
simple chain.

A broader class of priors for which inference is tractable is 
provided by Markov-tree distributions \cite{SchniterTree}. 
These are graphical models
that factors according to a tree graph (i.e. a graph without loops).
A cartoon of the resulting compressed sensing model is 
reproduced in Figure \ref{fig:TreePrior}.

The case of tree-structured priors is particularly relevant in imaging 
applications. Wavelet coefficients
of natural images are sparse (an important motivating remark
for compressed sensing) and non-zero entries tend to be localized
along edges in the image. As a consequence, they cluster in subtrees
of the tree of wavelet coefficients. A Markov-tree prior
can capture well this structure.

Again, reconstruction is performed exactly on the tree-structured prior
(this can be done efficiently using belief propagation),
while AMP is used to do inference over the compressed sensing measurements
(the upper part of Figure \ref{fig:TreePrior}).
%
%
\subsection{Sparse sensing matrices}

Throughout this review we focused for simplicity on dense measurement
matrices $A$. Several of the mathematical results presented in the previous
sections do indeed hold for dense matrices with i.i.d. 
components. Graphical models ideas are on the other hand
particularly useful for sparse measurements.

Sparse sensing matrices present  several 
advantages, most remarkably lower measurement and 
reconstruction complexities \cite{IndykCS}. While sparse constructions are
not suitable for all applications, they appear a promising
solution for networking applications, most notably in network
traffic monitoring \cite{Cormode,LuEtAl}.

\begin{figure}
\begin{center}
\includegraphics[width=3.75cm,angle=90]{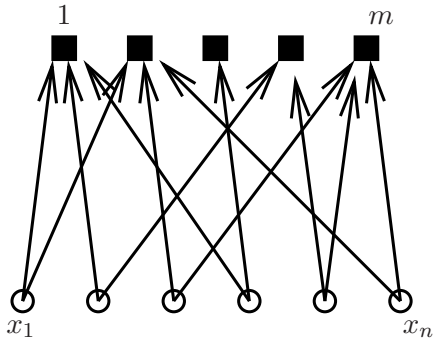}
\put(-154,-7){$x_1$}
\put(-4,-7){$x_n$}
\put(-135,110){$1$}
\put(-17,110){$m$}
\end{center}
\caption{{\small Sparse sensing graph arising in a networking application.
Each network flow (empty circles below) 
hashes into $k=2$ counters (filled squares).}}
\label{fig:Sparse}
\end{figure}
In an over-simplified example, one would like to monitor
the sizes of $n$ packet flows at a router.
It is a recurring empirical observation that most of the
flows consist of a few packets, while most of the traffic is 
accounted for by a few flows.
Denoting by $x_1$, $x_2$, \dots $x_n$ the flow sizes
(as measured, for instance, by the number of packets belonging to the flow),
it is desirable to maintain a small sketch of the
vector $x=(x_1,\dots,x_n)$. 

Figure \ref{fig:Sparse} describes a simple approach:
flow $i$ hashes into a small number --say $k$-- of memory spaces,
$\di = \{a_1(i),\dots, a_k(i)\}\subseteq [m]$. Each time a new packet
arrives for flow $i$, the counters in $\di$ are incremented.
If we let $y=(y_1,\dots , y_m)$ be the contents of the counters, we have
\begin{eqnarray}
y = Ax\, ,
\end{eqnarray}
where $x\ge 0$ and $A$ is a matrix with i.i.d. columns with $k$
entries per column equal to $1$ and all the other entries equal to $0$. 
While this simple scheme requires unpractically deep counters 
(the entries of $y$ can be large), \cite{LuEtAl} showed how to
overcome this problem by using a multi-layer graph.

Numerous algorithms were developed for compressed sensing reconstruction
with sparse measurement matrices
\cite{Cormode,XuHassibi,IndykCS,IndykCS2}. 
Most of these algorithms are based on 
greedy methods, which are essentially of message passing type.
Graphical models ideas can be used to construct such algorithms in
a very natural way. For instance,
the algorithm of \cite{LuEtAl} (see also \cite{LuAllerton1,LuAllerton2,Chandar}
for further analysis of the same algorithm) is closely related to the ideas presented in the rest of 
this chapter. It uses messages $x^{t}_{i\to a}$ (from variable nodes 
to function nodes) and $r_{a\to i}^t$ (from function nodes to variable nodes).
These are updated according to
\begin{eqnarray}
r^t_{a\to i} & = & y_a-\sum_{j\in\da\setminus i} x_{j\to a}^t\, ,\\
x^{t+1}_{i\to a} & = &
\left\{
\begin{array}{ll}
\min\big\{r^t_{b\to i} :\;\; b\in \di\setminus a\big\} &
\mbox{at even iterations $t$,}\\
\max\big\{r^t_{b\to i} :\;\; b\in \di\setminus a\big\} &
\mbox{at odd iterations $t$,}
\end{array}
\right.
\end{eqnarray}
where $\da$ denotes the set of neighbors of node $a$ 
in the factor graph.
These updates are very similar to Eqs.~(\ref{eq:mp-repeated}), 
(\ref{eq:mp-repeated-bis}) introduced earlier in our derivation of
AMP.

%
%
\subsection{Matrix completion}

`Matrix completion' is the task of inferring an (approximately) low rank 
matrix from observations on a small subset of its entries.
This problem has attracted considerable interest offer the last
two years due to its relevance in  a number of applied domains
(collaborative filtering, positioning, computer vision, etc.).

Significant progress has been achieved on the theoretical side.
The reconstruction question has been addressed in analogy with 
compressed sensing in
\cite{CaR09,CaP10,Gross09,Negahban09},
while  an alternative approach based on greedy methods was developed 
in \cite{KOM10,KOM10Noisy,KM11mp}.
While the present chapter does not treat 
matrix completion in any detail, it is interesting to
mention that graphical models ideas
can be useful in this context as well. 

\begin{figure}
\begin{center}
\includegraphics[width=3.25cm,angle=90]{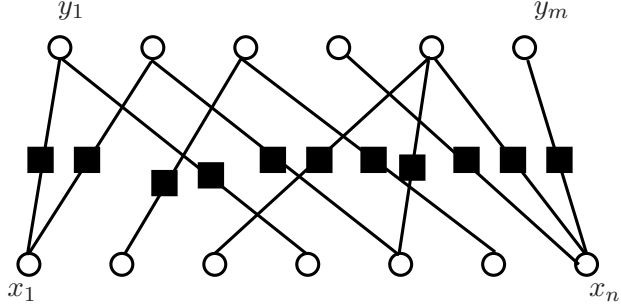}
\put(-224,-7){$x_1$}
\put(-4,-7){$x_n$}
\put(-205,100){$y_1$}
\put(-25,100){$y_m$}
\end{center}
\caption{{\small Factor graph describing the cost function
(\ref{eq:CostFunction}) for the matrix completion problem.
Variables $x_i,y_j\in\reals^r$ are to be optimized over.
The cost is a sum of pairwise terms (filled squares)
corresponding to the observed entries in $M$.}}
\label{fig:MatrixCompletion}
\end{figure}
Let $M\in\reals^{m\times n}$ be the matrix to be reconstructed, 
and assume that a subset $E\subseteq[m]\times [n]$ 
of its entries is observed. It is natural to try to accomplish
this task by minimizing the $\ell_2$ distance on observed entries.
For $X\in\reals^{m\times r}$, $Y\in\reals^{n\times r}$,
we introduce therefore the cost function
\begin{eqnarray}
\cost(X,Y) = \frac{1}{2}\,\|\cP_E(M-XY^T)\|_F^2
\end{eqnarray}
where $\cP_E$ is the projector that sets to zero the entries outside
$E$ (i.e. $\cP_E(L)_{ij}=L_{ij}$ if $(i,j)\in E$ and $\cP_E(L)_{ij}=0$ 
otherwise).
If we denote the rows of $X$ as $x_1,\dots, x_m\in\reals^r$
and the rows in $Y$  as $y_1,\dots, y_n\in\reals^r$,
the above cost function can be rewritten
as
\begin{eqnarray}
\cost(X,Y) = \frac{1}{2}\sum_{(i,j)\in E}\big(M_{ij}-\<x_i,y_j\>
\big)^2\, ,\label{eq:CostFunction}
\end{eqnarray}
with $\<\,\cdot\, ,\,\cdot,\>$ the standard scalar product on $\reals^r$.
This cost function factors accordingly the
bipartite graph $G$ with vertex sets $V_1=[m]$ and $V_2=[n]$
and edge set $E$. The cost decomposes as a sum of pairwise terms associated
with the edges of $G$.

Figure \ref{fig:MatrixCompletion} reproduces the graph $G$ that is 
associated to the cost function $\cost(X,Y)$. It is remarkable 
some properties of the reconstruction problem can be `read'
from the graph. For instance, in the simple case $r=1$,
the matrix $M$ can be reconstructed if and only if
$G$ is connected (banning for degenerate cases) \cite{KMO08}.
For higher values of the rank $r$, rigidity of the graph is
related to uniqueness of the solution of the reconstruction
problem \cite{Rigidity}.
Finally, message passing algorithms for this problem were studied
in \cite{IMP,KM11mp}.

%
%
\subsection{General regressions}

The basic  reconstruction method discussed in this 
review is the regularized least-squares regression
defined in Eq.~(\ref{eq:LASSO}), also known as the LASSO.
While this is by far the most interesting setting for
signal processing applications, for a number of
statistical learning problems, the linear model
(\ref{eq:FirstModel}) is not appropriate. Generalized linear
models provide a flexible framework to extend the ideas discussed here.

 An important example is logistic 
regression, which is particularly suited for the case 
in which the measurements $y_1,$ $\dots$ $y_m$ are $0$--$1$ valued. 
Within logistic regression, these are modeled
as independent Bernoulli random variables with
\begin{eqnarray}
p(y_a=1|x) = \frac{e^{A_a^Tx}}{1+e^{A_a^Tx}}\, ,
\end{eqnarray}
with $A_a$ a vector of `features' that characterizes the $a$-th 
experiment. The objective is to learn the vector $x$ of coefficients
that encodes the relevance of each feature.
A possible approach consists in minimizing the regularized 
(negative) log-likelihood, that is
\begin{align}
\cost^{{\rm LogReg}}_{A,y}(z) & \equiv -\sum_{a=1}^m y_a (A_a^Tz)+
\sum_{a=1}^m\log\big(1+e^{A_a^Tz}\big)+  \lambda
\|z\|_1\, ,
\end{align}
The papers \cite{RanganGen,BayatiLogReg} develops 
approximate message passing 
algorithms for solving optimization problems of this type.

%
%

\section*{Acknowledgements}

It is a pleasure to thank  
Mohsen Bayati, Jose Bento,
David Donoho and Arian Maleki,
with whom this research has been developed.
This work was partially supported by a Terman fellowship,
the NSF CAREER award CCF-0743978 and the NSF grant DMS-0806211.

%
%
\newpage

%
%
\bibliographystyle{amsalpha}

\newcommand{\etalchar}[1]{$^{#1}$}
\providecommand{\bysame}{\leavevmode\hbox to3em{\hrulefill}\thinspace}
\providecommand{\MR}{\relax\ifhmode\unskip\space\fi MR }
\providecommand{\MRhref}[2]{%
  \href{http://www.ams.org/mathscinet-getitem?mr=#1}{#2}
}
\providecommand{\href}[2]{#2}

\end{document}